\documentclass[11pt,english]{article}

\usepackage[margin=1in]{geometry}

\usepackage[utf8]{inputenc} 
\usepackage[T1]{fontenc}
\usepackage{amsmath,amsthm,mathtools,amssymb}
\usepackage{thmtools,thm-restate}
\usepackage{algorithm}
\usepackage{algorithmic}

\usepackage[shortlabels]{enumitem}
\usepackage[usenames,svgnames,xcdraw,table]{xcolor}
\definecolor{DarkBlue}{rgb}{0.1,0.1,0.5}
\definecolor{DarkGreen}{rgb}{0.1,0.5,0.1}
\usepackage{hyperref}
\hypersetup{
   colorlinks   = true,
 linkcolor    = DarkBlue, 
 urlcolor     = DarkBlue, 
	 citecolor    = DarkGreen 
}
\usepackage[capitalize]{cleveref}
\usepackage{tikz}
\usepackage{graphicx}
\usepackage{svg}
\usepackage{float}
\usepackage{verbatim}
\usepackage{caption}
\usepackage{natbib}
\usepackage{fancyhdr}
\usepackage{booktabs}

\usepackage[many]{tcolorbox}  

\newtheorem{theorem}{Theorem}
\newtheorem{claim}{Claim}

\newtheorem{observation}{Observation}
\newtheorem{example}{Example}
\newtheorem{remark}{Remark}
\newtheorem{proposition}{Proposition}

\newcounter{casenum}

\DeclareMathOperator*{\argmax}{arg\,max}
\DeclareMathOperator*{\argmin}{arg\,min}

\newlist{todolist}{itemize}{2}
\setlist[todolist]{label=$\square$}
\usepackage{pifont}

\newcommand{\EF}{\textsc{EF}}
\newcommand{\EQ}{\textsc{EQ}}
\newcommand{\EFone}{\textsc{EF1}}

\newcommand{\PO}{\textsc{PO}}
\newcommand{\SO}{\textsc{SO}}
\newcommand{\NT}{\textsc{TS}}
\newcommand{\wNT}{\textsc{wTS}}

\newcommand{\NPH}{\textsc{NP-hard}}

\makeatletter
\newcommand{\NoNumber}[1]{%
  \let\@oldalglinenumber\alglinenumber%
  \renewcommand{\alglinenumber}{}%
  \STATE #1%
  \let\alglinenumber\@oldalglinenumber%
}
\makeatother

\DeclarePairedDelimiter\ceil{\lceil}{\rceil}
\DeclarePairedDelimiter\floor{\lfloor}{\rfloor}

\begin{document}
\author{}
\date{}

\title{Non-Monotonicity in Fair Division of Graphs}

\author{
    Hadi Hosseini \\ 
    \small Penn State University\\ 
    \small \texttt{hadi@psu.edu}
    \and 
    Shraddha Pathak\\
    \small Penn State University\\ \small \texttt{ssp5547@psu.edu}
    \and
    Yu Zhou\footnote{Corresponding author. This work was done when Yu Zhou was a Postdoc at Penn State University.}\\
    \small Beijing Normal University\\ 
    \small \texttt{yu.zhou@bnu.edu.cn} 
}
\maketitle

\begin{abstract}

We consider the problem of fairly allocating the vertices of a graph among $n$ agents, where the value of a bundle is determined by its \textit{cut value}---the number of edges with exactly one endpoint in the bundle.
This model naturally captures applications such as team formation and network partitioning, where valuations are inherently \textit{non-monotonic}: the marginal values may be positive, negative, or zero depending on the composition of the bundle.
We focus on the fairness notion of envy-freeness up to one item (\EFone{}) and explore its compatibility with several efficiency concepts such as Transfer Stability (\NT) that prohibits single-item transfers that benefit one agent without making the other worse-off.
For general graphs, our results uncover a non-monotonic relationship between the number of agents $n$ and the existence of allocations satisfying \EFone{} and transfer stability (\NT): such allocations always exist for $n=2$, may fail to exist for $n=3$, but exist again for all $n\geq 4$. 
We further show that existence can be guaranteed for any $n$ by slightly weakening the efficiency requirement or by restricting the graph to forests. All of our positive results are achieved via efficient algorithms.

\end{abstract}


 \section{Introduction}

Fair division of indivisible items is a fundamental problem in multiagent systems and computational social choice.
A central goal in fair division is to allocate a set of indivisible, non-shareable items among (potentially heterogeneous) agents in a manner that is both fair and efficient, guided by well-established axiomatic principles.
However, much of the existing literature assumes that agents have additive, or otherwise \textit{monotonic} valuations.
These assumptions often fail to capture more realistic scenarios where the value of an item is non-monotonic and may depend on the composition of the bundle it belongs to.
%

In this paper, we initiate a formal study of non-monotonic \textit{bundle-dependent} valuations, where the marginal value of an item may vary---potentially becoming negative---depending on the current bundle held by the agent.
%
%
For instance, in settings such as team formation or committee selection with diverse expertise, the value of an additional member, such as a goalkeeper, may vary significantly depending on the existing composition: they may be essential for a team lacking a substitute, but redundant or even burdensome for a fully staffed team due to added financial cost.

We focus on a specific class of bundle-dependent preferences based on \textit{cut-valuations} defined over undirected graphs.
In this model, items correspond to vertices of a graph $G=(V,E)$, and an agent's value for a bundle is determined by the number of edges that are cut, i.e. edges with exactly one endpoint in the bundle.
Thus, the marginal value of an item is positive for bundle $S$ if it has less neighbors inside $S$ than outside $S$, otherwise it is zero or negative. 
This yields valuations that are inherently \textit{non-monotonic}: adding an item can increase or decrease an agent's utility depending on the underlying structure of the bundle. 

We are interested in allocating vertices of the graph (corresponding to items) in a fair way while guaranteeing some notion of economic efficiency among the agents.
Cut-valuations naturally model the aforementioned team formation problems that value fairness and diversity, where vertices represent individuals in a social or similarity graph and edges capture social ties or overlapping expertise.
These non-monotonic valuations can also be seen as a soft relaxation of conflict constraints commonly studied in fair division \citep{DBLP:journals/aamas/HummelH22,chiarelli2020fair,DBLP:conf/nips/LiLZ21}. Such interplay between constraint-based models and their valuation-based counterparts have been previously considered for matroids \citep{DBLP:conf/aaai/BiswasB19,cookson2025constrained,barman2020existence}, as well as cardinality constraints and level valuations \citep{DBLP:conf/ijcai/BiswasB18,christodoulou2024fair}.

 \subsection{Our Contributions}

We consider the problem of partitioning general graphs, as well as forests, with a focus on simultaneously achieving fairness and economic efficiency.
For fairness, we consider \textit{envy-freeness} (\EF{}), which informally requires that no agent strictly prefers another agents' bundle; we also consider its well-studied relaxation, \textit{envy-freeness up to one} (\EFone{}), which requires that any pairwise envy is eliminated by removal of a single item.
Under cut-valuations, allocating all vertices to a single agent is vacuously envy-free: all agents receive a value of zero.
To avoid such degenerate outcomes, we study fairness in conjunction with efficiency, including: (i) \emph{social optimality} (\SO), which maximizes total utility;
(ii) \emph{Pareto optimality} (\PO), where no agent can be made better off without harming another.
We also propose two natural and desirable efficiency notions:
(iii) \emph{transfer stability} (\NT), prohibiting single-item transfers that benefit one agent without making the other worse-off; 
and (iv) \emph{weak transfer stability} ($\wNT$), forbidding transfers that strictly benefit both agents involved.\footnote{In \Cref{sec:prev_tech_not_work}, we demonstrate that techniques developed for monotone valuations fail to extend to the cut-valuation setting.}

\paragraph{General graphs.}
Our first set of results reveal an axiomatic non-monotonic relationship between existence guarantees for \EFone{}+\NT{} allocation and the number of agents $n$ under cut-valuations: They exist for $n=2$ (\Cref{prop:n=2}), may fail to exist for $n=3$ (\Cref{thm:general:non_exist}), but exist again for all $n\geq 4$ (\Cref{thm:general:NT4}).
The $n=3$ counter-example also shows the incompatibility of $\EFone$ with $\PO$ and $\SO$. 
In contrast, the existence of $\EFone$ and $\SO$ allocations follows a monotonic pattern: They always exist for $n=2$, but may not exist when $n\geq 3$ (\Cref{thm:general:EF1_SO}).

This raises a natural question: \textit{Is there an efficiency notion that is always compatible with \EFone{}, regardless of $n$?}
We show that a slight relaxation of \NT{} (\wNT{}) is the strongest efficiency notion which can always be guaranteed alongside \EFone{}, irrespective of $n$ (\Cref{thm:general:ne}). 
Moreover, such allocations can be computed in polynomial time.
We also show that allocations that are $\frac{1}{2}$-\EFone{}, a multiplicative approximation of \EFone{}, can co-exist alongside \PO{} for arbitrary $n$ (\Cref{thm:approximate_EF1_PO}).  

We note that our $\EFone$+$\wNT$ result (\Cref{thm:general:ne}) also improves upon prior guarantees for equitable graph cut. Specifically, we show that the vertices of a graph can be partitioned into $n$ parts such that the absolute difference between the cut-values of any two parts is bounded by the maximum degree $\Delta$ in the graph (\Cref{cor:equitable_cuts}). This improves on the previous bound of $5\Delta + 1$ by \citet{barman2025fair}. 

\paragraph{Forests.}
When the underlying graph is a forest (i.e. acyclic), we show that \EFone{} and \SO{} allocations always exist and can be computed in polynomial time (\Cref{thm:forest}). Interestingly, the standard techniques from the monotone settings---iteratively building toward an \EFone{} allocation where, on assigning each new item, \EFone{} is maintained as an invariant---do not apply to such special cases either. While our algorithm also iteratively builds towards an \EFone{} allocation, there are situations where we temporarily break intermediate $\EFone$ guarantees, but ensure that such violations can be `caught-up' through subsequent assignments.

Finally, we note that the negative counter-example we construct for showing the incompatibility of \EFone{} alongside various fairness notions under general graphs is in fact a \emph{complete bipartite graph} with treewidth $2$. 
This highlights that while strong existence and algorithmic guarantees hold for forests (bipartite graphs; treewidth 1), they break even for graphs that are only slightly different.

\subsection{Related works} \label{sec:related_works}
\paragraph{Fair Division under Non-monotone Valuations}
Early works on non-monotone valuations focus on \emph{doubly monotone} (often \emph{additive}) valuations, where each item is either always a good (positive value) or a chore (negative value) for an agent, regardless of the bundle.  
In such settings, \EFone{} allocations always exist \citep{DBLP:journals/aamas/AzizCIW22,bhaskar2020approximate}.
The existence of \EFone{}+\PO{} allocations remains open, even in the simple additive case \citep{liu2024mixed}.  
While such allocations are known for (additive) goods-only instances \citep{caragiannis2019unreasonable}, their existence in chores-only settings was only recently established \citep{mahara2025existence}. 
Recent work either relaxes \EFone{} \citep{barman2025fairandefficient} or focuses on specific valuations \citep{hosseini2023fairly}.
More general non-monotone settings---where items act as goods or chores depending on the bundle---have been recently explored under special classes such as subadditive, non-negative valuations\footnote{We note that cut-valuations studied in this paper are a subclass of non-negative and were discussed in \citep{barman2025fair}.}, identical trilean, and separable single-peaked \citep{barman2025fair,bilo2025almost,bhaskar2025towards}.

\paragraph{Fair Division of Graphs}
Graphs appear in fair division both as feasibility constraints (e.g. conflicting or connectivity \citep{DBLP:journals/aamas/HummelH22,chiarelli2020fair,DBLP:conf/nips/LiLZ21,kumar2024fair,igarashi2025dividing,DBLP:conf/ijcai/BouveretCEIP17,bei2022price,bilo2022almost}), and to define agents' valuations (e.g. from graphical structures such as shortest paths \citep{hosseini2025fair,hosseini2025algorithmic}, minimum vertex cover \citep{DBLP:conf/www/WangL24}, or edge-based models where each item (an edge) is valued by exactly two agents (the incident vertices) \citep{christodoulou2023fair,deligkas2024ef1,zhou2024complete}).
A related line of research considers partitioning of friends into groups that are balanced and fair \citep{li2023partitioning,deligkas2025balanced}. Here, an agents utilities are monotone and based on intra-group connections. In contrast, our model captures non-monotonic, bundle-dependent preferences, arising in problems where inter-group connections are valued.

\section{Preliminaries}

Given any $j,k\in \mathbb{N}$ such that $j\leq k$, let $[k] := \{1,\ldots, k\}$, and $[j,k] := \{j, j+1, \ldots, {k}\}$. 

An instance of the cut-valuation problem is denoted by $\langle N, G, v \rangle$, where $N = \{1, 2, \ldots, n\}$ is the set of $n$ agents, $G = (V, E)$ is an undirected graph with $|V|=m$ vertices and $|E|$ edges, and a cut-valuation $v$. The vertex set $V$ corresponds to the set of items to be distributed among the agents; throughout the paper, we use items and vertices interchangeably. 
The edge between any two adjacent vertices $o, o^\prime$ is denoted by $e = (o,o')$. 
The neighbors of $o$ within the set $S$ is denoted by $N_S(o) : = \{o'\in S :\ (o, o') \in E\}$. The degree of $o$, denoted by $\deg(o) := |N_{V}(o)|$, is the number of its neighbors in $G$.
A graph is a forest if it is acyclic. Note that a graph may contain multiple disconnected components.

\paragraph{Cut-Valuations.}
The cut-valuation $v : 2^V \to \mathbb{R}$ defines the cardinal preferences of the agents for any subset of items in $V$. 
Specifically, the cut-value of any subset $S$ is $v(S) \coloneq |\{ (o,o')\in E: o\in S\ \text{and}\ o'\in V \setminus S\}|$, which is the number of edges that have one endpoint in $S$ and the other endpoint in $V\setminus S$. 
Crucially, the cut-valuation represents a subclass of \textit{non-monotone} valuations, where an item $o$ has a \textit{positive marginal value} for a bundle $S$ if and only if $o$ has more neighbors outside of $S$ than in $S$, i.e. $|N_S(o)| < |N_{V\setminus S}(o)|$.
Otherwise, $o$ may have a \textit{negative or non-positive marginal value}. 
For the ease of exposition, we refer to an item with positive marginal value (with respect to a subset) as a \textit{good}, otherwise it is called a  \textit{chore} (with respect to the subset).
Without loss of generality, we assume (i) $|V|>n$, and (ii) $G$ has no singleton components because such vertices have zero value and can be assigned arbitrarily.

\paragraph{Allocation.}

An allocation $A = (A_1, \ldots, A_n)$ is a partition of vertices (items) into $n$ pairwise disjoint subsets of $V$, i.e., for every $A_i,A_j \subseteq V$, $A_i \cap A_j = \emptyset$.
Each part $A_i$ contains a bundle of items assigned to agent $i\in N$. 
Note that the items in $A_i$ may not be adjacent. 
An allocation may be \emph{partial} when $\cup_{i \in N} A_i \subsetneq V$, or \emph{complete} when $\cup_{i \in N} A_i = V$. 
Let us review the cut-valuations and potential allocations using \Cref{ex:cutvalue}. 

\begin{example} [Allocations and cut-valuations]
\begin{figure}[htp]
\centering
\begin{tikzpicture}[every node/.style={circle, draw=none, minimum size=5mm, inner sep=1pt, font=\scriptsize}]

    \node (x) at (0,0) {$o_1$};
    \node (a) at (-1.2,-1.5) {$o_2$};
    \node (b) at (0,-1.5) {$o_3$};
    \node (c) at (1.2,-1.5) {$o_4$};
    
    \draw (x) -- (a);
    \draw (x) -- (b);
    \draw (x) -- (c);

    \node (y) at (3,0) {$o_5$};
    \node (d) at (1.8,-1.5) {$o_6$};
    \node (e) at (3,-1.5) {$o_7$};
    \node (f) at (4.2,-1.5) {$o_8$};

    \draw (y) -- (d);
    \draw (y) -- (e);
    \draw (y) -- (f);

    \draw (x) -- (y);

    \draw[blue, thick, dotted, rounded corners] (-0.35,0.35) rectangle (0.35,-0.35) node[above right, xshift=1pt, yshift=1pt] {\scriptsize $S$};

    \draw[red, thick, dashed, rounded corners] (-0.3,0.3) rectangle (0.3,-1.75);
    \node[font=\scriptsize] at (-0.5,-0.9) {\textcolor{red}{$T$}};

    \draw[green!60!black, thick, dashed, rounded corners] (-1.5,-1.2) rectangle (0.38,-1.85);
    \node[font=\scriptsize] at (-1.7,-1.3) {\textcolor{green!60!black}{$U$}};

    \draw[magenta!70!white, thick, dotted, rounded corners] (-1.4,-1.25) rectangle (-0.75,-1.8);
    \node[font=\scriptsize] at (-1.4,-2.1) {\textcolor{magenta}{$W$}};

\end{tikzpicture}
\caption{Graph with $8$ items (vertices). 
}
\label{fig:cut_value_example}
\end{figure}
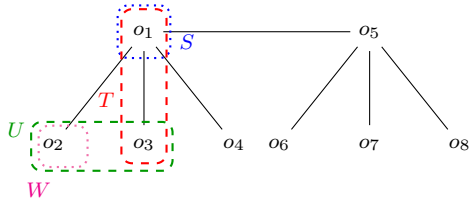

Consider the graph in \Cref{fig:cut_value_example}. The cut-value of the set $S = \{o_1\}$ is $v(S)=4$. However, the set $T$ that is obtained by adding an item $o_3$ to $S$ reduces the cut-value to $v(T)=3$.
In other words, $o_3$ is a chore for the set $S$. However, $o_3$ is a good for set $W=\{o_2\}$ since addition of $o_3$ to $W$ results in the set $U$ that has a cut-value of $v(U)=2$, whereas $v(W)=1$. 
\label{ex:cutvalue}
\end{example}

\begin{remark}[Relation to Graph Partitions]
In graph theory, a $k$-partition of a graph $G = (V, E)$ is a division of $V$ into $k$ disjoint subsets $V = V_1 \cup \cdots \cup V_k$, where $V_i \cap V_j = \emptyset$ for all $i \ne j$, where $V_i$ may be subject to additional structural constraints such as being connected subgraphs or forming independent sets.
In our model, an allocation corresponds to a complete partition of $V$ into $n$ bundles $(A_1, \ldots, A_n)$, one for each agent. 
Our framework \emph{does not} impose a connectivity constraint, i.e. multiple disconnected components may be allocated as a bundle (say $A_i$) to an agent.
\end{remark}

\paragraph{Fairness.}
An allocation $A = (A_1, \ldots, A_n)$ is \textbf{envy-free} (EF) if every agent values their own bundle at least as much as that of any other agent's bundle, i.e., $v(A_i) \geq v(A_j)$ for every two agents $i, j \in N$.
An allocation $A$ is \textbf{envy-free up to one item} (\EFone{}) if for every pair of agents $i,j\in N$ where $v(A_j) > v(A_i)$, there exists an item $o_j\in A_j$ such that $v(A_i) \geq v(A_j \setminus \{o_j\})$\footnote{We refer the reader to \Cref{sec:prev_tech_not_work} for a discussion on alternative definitions of \EFone{} used in the literature.} \citep{lipton2004approximately,budish2011combinatorial}.

In \cref{ex:cutvalue}, although $v(T) > v(W)$, the (hypothetical) removal of $o_1\in T$ eliminates the envy. Note that in this setting, the existence of an \EFone{} allocation is particularly delicate due to the non-monotonicity of valuations since removing a single item (vertex) may sometimes increase the agent's value. 
In this example, the removal of $o_3\in T$ in fact increases the value of $T$ from 3 to 4.

\paragraph{Efficiency.}
In our setting, fairness alone may not be sufficient to rule out undesirable allocations. For instance, allocating all vertices to a single agent is envy-free, however, all agents receive a value of zero.
Thus, we consider \textit{non-empty} allocations wherein each agent receives at least one item and define several plausible efficiency notions.

An allocation, $A$, is \textbf{Social Optimal} (\SO{}) if it maximizes the utilitarian social welfare, i.e., $A \in \argmax_{X \in \Pi_n(V)}\sum_{i \in N}v(X_i)$, where $\Pi_n(V)$ denotes the set of all $n$-partitions of $V$.
An allocation $A$ is \textbf{Pareto Optimal} (\PO{}) if it is not Pareto dominated by another allocation, that is, no other allocation $B$ exists such that $v(B_i) \ge v(A_i)$ for every $i \in N$ with at least one inequality being strict.


An allocation, $A = (A_1, \ldots, A_n)$, is \textbf{Transfer Stable} (\textsc{TS}) if, for every $i, j\in N$, there does not exist an item $o\in A_i$ such that $v(A_i \setminus \{o\}) \geq v(A_i)$ and $v(A_j \cup \{o\}) \geq v(A_j)$ with at least one inequality being strict. 
Allocation $A$ is \textbf{Weak Transfer Stable} (\textsc{wTS}) if, for every $i, j\in N$, there does not exist an item $o\in A_i$ such that $v(A_i \setminus \{o\}) > v(A_i)$ and $v(A_j \cup \{o\}) > v(A_j)$.
Informally, a $\NT$ allocation requires that no single item can be transferred from an agent to another without making either agent worse-off while increasing the value of at least one of the two agents. 
It is \wNT{} if no such transfer exists that necessarily makes both agents strictly better-off.
It is easy to see that $\SO \implies \PO \implies \NT \implies \wNT$ while the converse directions do not hold; see \Cref{ex:efficiency_notions}.

\begin{example}[Understanding Efficiency Notions]\label{ex:efficiency_notions}
Let us revisit the example given in \Cref{fig:cut_value_example} and consider 7 agents. 
The allocation $A=(\{o_1, o_6\}, \{o_2\}, \{o_3\}, \{o_4\}, \{o_5\}, \{o_7\}, \{o_8\})$ is $\SO$. The allocation $(\{o_1, o_5\}, \{o_2\}, \{o_3\}, \{o_4\}, \{o_6\}, \{o_7\}, \{o_8\})$ is $\PO$ but not $\SO$. 
For an example of an allocation that is $\NT$ but not $\PO$, consider the same graph but let there be $n=4$ agents. The allocation $(\{o_1, o_5\}, \{o_2, o_6\}, \{o_3, o_7\}, \{o_4, o_8\})$ satisfies $\NT$ but is not $\PO$ since it is Pareto dominated by the allocation $(\{o_1, o_6, o_7\}, \{o_5\}, \{o_2, o_3\}, \{o_4, o_8\})$. 
Finally, to illustrate an allocation that is $\wNT$ but not $\NT$, consider a cycle with $6$ nodes ($o_1, o_2, \ldots, o_6$) and $n=3$ agents. The allocation in which two consecutive nodes (say $(\{o_1, o_2\}, \{o_3, o_4\}, \{o_5, o_6\})$ are assigned to the same agent is $\wNT$ but not $\NT$.
\end{example}

\section{General Graphs}\label{sec:general_graphs}

In this section, we study general graphs that may contain cycles. 
Our first set of results reveal a non-monotonic relationship between existence guarantees for \EFone{}+\NT{} allocation and the number of agents under cut-valuations: They exist for $n=2$, may fail to exist for $n=3$, but exist again for all $n\geq 4$. Interestingly, this non-monotonic dependence is not present for \SO{}. In \Cref{sec:EF1_SO_example}, we provide a family of instances for every $n\ge 3$ agents such that no allocation is simultaneously \EFone{} and \SO{}.

\subsection{Two Agents: Existence}
To warm up, we start by considering the case of $n=2$, and give an existence guarantee for \EF{} and \SO{}: Since the cut function of a graph $G$ is symmetric, i.e. $v(S) = v(V\setminus S)$ for any $S\subseteq V$, \textit{any} partition between two agents is \EF{}. 
Moreover, an allocation is \SO{} if and only if it is \PO{}, because the cut-values are symmetric and an increase in the value of one agent leads to an increase in value for both. 

\begin{restatable}{proposition}{twoagents}\label{prop:n=2}
When $n=2$, an allocation satisfying \EF{} and \SO{} always exists, but computing such an allocation is \NPH{}.
Moreover, an allocation satisfying \EF{} and \NT{} always exists and can be computed in polynomial time. 
\end{restatable}

\begin{proof}
     For two agents, every allocation is \EF{} due to symmetry of cut-valuations. 
     Thus, the problem reduces to finding an \SO{} partition.
     The symmetry of valuations means that the partition corresponding to the maximum cut is \SO{}. That is, a \textsc{max-cut} results in \EF{}+\SO{}; however,
     since computing a \textsc{max-cut} is \NPH{} \citep{garey1974some}, the hardness follows.
     
     Focusing on \NT{}, we devise a simple greedy algorithm as follows:
     Start with an arbitrary $2$-partition of the graph, say $A = (V, \emptyset)$, and iteratively transfer an item from one bundle to the other such that the value of the bundles increases. 
     The value can increase at most $O(m^2)$ times since, the value increases by at least $1$ in each round and for any set $S$, $v(S) \leq |E|$ where $|E|\in O(m^2)$ for simple graphs. 
     Again, by symmetry of cut-valuations for 2-partitions, the final allocation after the greedy process is also \EF{}.
\end{proof}

\subsection{Three Agents: Non-existence} \label{sec:general_graphs_general}

While for $n=2$, \EF{}+\NT{} always admits a polynomial-time solution, we show that for $n = 3$, even \EFone{}---a strictly weaker notion than \EF{}---is incompatible with \NT{} (and thus, with \PO{} and \SO{} as well).
%

\begin{restatable}[Non-existence of \EFone{}+\NT{} allocations]{theorem}{nonexist}
\label{thm:general:non_exist}
    There exists a cut-valuation instance with $n=3$ agents where no \EFone{} allocation is \NT{} (and thus \PO{} and \SO{}). 
\end{restatable}

\begin{proof}
\begin{figure}[htp]
    \centering
    
    \begin{tikzpicture}
    \small
    \node[circle, draw=red, thick, minimum size=0.8cm] (Top) at (0, 2) {$o_a$};
    
    \node[circle, draw=blue, thick, minimum size=0.8cm] (Bottom) at (0, 0) {$o_b$};

    \node[circle, draw=yellow, thick, minimum size=0.8cm] (C1) at (-3, 1) {$o_{c_1}$};
    \node[circle, draw=yellow, thick, minimum size=0.8cm] (C2) at (-1.5, 1) {$o_{c_2}$};
    \node[circle, draw=red, thick, minimum size=0.8cm, inner sep=0pt] (C5) at (1.5, 1) {$o_{c_{d-1}}$};
    \node[circle, draw=red, thick, minimum size=0.8cm] (C6) at (3, 1) {$o_{c_d}$};
    
    \node at (0, 1) {$\cdots$};

    \draw[-] (Top) -- (C1);
    \draw[-] (Top) -- (C2);
    \draw[-] (Top) -- (C5);
    \draw[-] (Top) -- (C6);

    \draw[-] (Bottom) -- (C1);
    \draw[-] (Bottom) -- (C2);
    \draw[-] (Bottom) -- (C5);
    \draw[-] (Bottom) -- (C6);
\end{tikzpicture}
\caption{An example showing the incompatibility of \EFone{} and transfer stability (\NT{}) for $n=3$.}
\label{fig:nonexistence}
\end{figure}
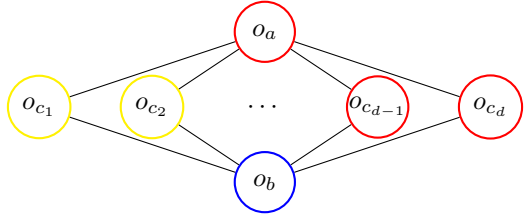

Consider the instance with $3$ agents and $d+2$ items, $o_a, o_b$ and $o_{c_t}$ where $t\in[d]$ as shown in \Cref{fig:nonexistence}.
Suppose $d \geq 3$ and that $d$ is odd. In any $\EFone$ allocation, $o_a$ and $o_b$ must be allocated to different agents; if not, i.e., if $A_1 = \{o_a, o_b\}$, then one of the remaining two agents will receive a value of at most $2\floor*{\frac{d}{2}} < d = v(A_1 \setminus \{o_a\}) = v(A_1 \setminus \{o_b\})$, which violates $\EFone{}$. Thus, $o_a$ is allocated to agent 1 (red) and $o_b$ to agent 2 (blue).

After allocating $\ceil*{\frac{d}{2}}$ of the remaining items to the third agent (yellow), no new item can be allocated to $A_3$ without breaking \EFone{} because each item has marginal value of two for agent 3, and on addition of any new item $o\in \{ o_{\ceil*{\frac{d}{2}}+1}, \ldots, o_{c_d} \}$, we have $v(A_3 \cup \{o\} \setminus \{o'\}) = 2\lceil\frac{d}{2}\rceil > d = v(A_1) = v(A_2)$ for any $o'\in A_3\cup \{o\}$. Thus, all the remaining items must be (wastefully) allocated to either agent $1$ or $2$ in order to maintain \EFone{}. In \Cref{fig:nonexistence}, they are arbitrarily allocated to agent 1 (red).
But the marginal values of these items are 0 for both agents 1 and 2; thus, to maintain \NT{} these items should be transferred to agent 3, which whereas violates \EFone{}. 
\end{proof}  

While \Cref{thm:general:non_exist} shows the incompatibility of \EFone{} with \NT, \PO, and \SO{}, the \EFone{} allocation marked in \Cref{fig:nonexistence} satisfies the weaker \wNT{} condition. In fact, as we show in \Cref{sec:arbitrary_n_positive_results}, allocations that are both \EFone{} and \wNT{} always exist and can be computed in polynomial time.

\subsection{Four or More Agents: Existence}\label{sec:four_and_more}

Surprisingly, when $n\ge 4$, \EFone{} and \NT{} allocations always exist and can be computed in polynomial time. 

\begin{restatable}[\EFone{}+\NT{} allocations when $n\ge 4$]{theorem}{thmEFTSnfour}\label{thm:general:NT4}
When $n \ge 4$, a complete allocation satisfying \EFone{} and \NT{} always exists and can be computed in polynomial time.
\end{restatable}

Before presenting the main algorithm for our constructive proof, we observe that, under identical valuations, 
if the allocation is ordered from the least to the highest valued bundle, then the allocation is \EFone{} \textit{if and only if} the agent with the least-valued bundle is not involved in any \EFone{} violation; we refer the readers to \Cref{sec:proofs_four_or_more} for the proof.
This lemma enables us to restrict attention to \EFone{} violations involving only the agent with the least-valued bundle.

\begin{restatable}{lemma}{identicalsufficesagentone}\label{lem:identical}
    Given an instance with identical valuations and allocation $A=(A_1, \ldots, A_n)$ such that $v(A_1) \leq \ldots \leq v(A_n)$, $A$ is \EFone{} if and only if, for every $j\in N$ such that $v(A_j) > v(A_1)$, there exists an item $o_j\in A_{j}$ s.t. $v(A_1) \geq v(A_j \setminus \{o_j\})$. 
\end{restatable}

\Cref{alg:general:NT} operates over the space of complete allocations and performs local search to eliminate \EFone{} violations. A subroutine (\Cref{alg:general_NT_subroutine}) is repeatedly invoked to enforce \NT{} by reassigning items whose marginal value to the current holder is weakly negative.

\begin{algorithm}[htp]
\begin{algorithmic}[1]
\small
\STATE Initialize $A = (A_1, \ldots,A_n)$ arbitrarily with any complete allocation.
\STATE Relabel bundles s.t. $v(A_1) \le \cdots \le v(A_n)$. 
\STATE $A\gets \textsc{TS-Subroutine}(\text{None}; A)$ 
\WHILE{($\exists$ $\EFone$ violation from agent $1$ to some other agent)}\label{line:EF1_while_NT}
\begin{tcolorbox}[colback=gray!5!white,colframe=white!75!black]\scriptsize
 \COMMENT{Case I (Lines \ref{line:case1_NT}--\ref{line:case1_NT_end}): When there exists an item in an \EFone{}-violator's bundle with positive marginal value for $A_1$. }
 \end{tcolorbox}
\WHILE{(there exists \EFone{} violation towards an agent $i$, s.t. $v(A_1 \cup \{o\}) > v(A_1)$ for some $o\in A_i$)} \label{line:case1_NT}
    \STATE Transfer item $o$ to $A_1$
    \STATE Relabel bundles s.t. $v(A_1) \le \cdots \le v(A_n)$
    \STATE $A \gets \textsc{TS-Subroutine}(\text{None}; A)$ 
    \label{line:NT-subroutine-1}
\ENDWHILE\label{line:case1_NT_end}
 \begin{tcolorbox}
[colback=gray!5!white,colframe=white!75!black]\scriptsize
 \COMMENT{Case II (Lines \ref{line:case2_NT}--\ref{line:NT-subroutine-2}): When all items in the \EFone{}-violator's bundle have non-positive marginal value for $A_1$. }
 \end{tcolorbox}
    \WHILE{(there exists \EFone{} violation from agent $1$ to 
    special agent $i^{*}$ and 
    $v(A_1) \geq v(A_1 \cup \{o\})$ for every $o\in A_{i^*}$)} \label{line:case2_NT}
    \STATE Transfer any such item $o$ to another agent, $k\neq i^{*}$, s.t. $v(A_k \cup \{o\}) > v(A_k)$ 
      
\ENDWHILE \label{line:case2_NT_end}

\STATE Relabel bundles s.t. $v(A_1) \le \cdots \le v(A_n)$.
    \STATE $A \gets \textsc{TS-Subroutine}(i^*; A)$ 
\label{line:NT-subroutine-2}
\ENDWHILE
\end{algorithmic}
\caption{Computing $\EFone$ + $\textsc{TS}$ allocations on general graphs when $n\ge 4$
} 
\label{alg:general:NT}
\end{algorithm}

\paragraph{Algorithm Description.}
\Cref{alg:general:NT} starts with any arbitrary complete allocation.
Since agents have identical valuation functions, any bundle can be assigned to any agent. Thus, throughout the algorithm, the bundles are sorted in non-decreasing order of value i.e. $v(A_1) \le v(A_2) \le \ldots \le v(A_n)$. 
The algorithm proceeds by iteratively resolving pairwise \EFone{} violations in two cases: 
In \textbf{Case I}, we run a \textit{local search} repeatedly to find a bundle, say $A_i$, such that 
i) $A_i$ violates \EFone{} with respect to the lowest valued bundle $A_1$, and
ii) there exists an item $o\in A_i$ with strictly positive marginal value for agent 1, i.e. $v(A_1 \cup \{o\}) > v(A_1)$. In this case, we transfer the item to $A_1$, update the labels, and repeat.
Otherwise, we are in the case (\textbf{Case II}) where agent 1's violation of \EFone{} towards another bundle, say $A_{i^*}$, cannot be remedied by transferring an item from $A_{i^*}$ to $A_1$ since all items in $A_{i^*}$ are of a non-positive marginal value for agent $1$. In this case, the algorithm transfers an item $o'\in A_{i^*}$ to another bundle, say $A_k$, for whom $v(A_{k} \cup \{o'\}) > v(A_k)$. 
This process repeats until there is no \EFone{} violations from agent $1$ to $i^*$.
The key observation is that when the algorithm is in Case II, there is an \EFone{} violation with respect to only a single agent (see \Cref{ob:general:ne:case2} below); thus, it suffices to consider only a single special agent $i^*$.
Note that to maintain efficiency, after each transfer, the algorithm invokes the \NT{} subroutine (\Cref{alg:general_NT_subroutine}) described next.

\paragraph{Maintaining Efficiency \NT{}-Subroutine.}
This subroutine (\Cref{alg:general_NT_subroutine}) takes as input a complete allocation and a designated \textit{special agent} (call it $i^*$), and returns a \NT{} allocation in which no agent is worse-off compared to the input allocation.
The special agent is used solely to ensure that, in certain cases, its bundle remains unchanged; which becomes critical when arguing the convergence of the algorithm (see the next paragraph). 
The subroutine iteratively identifies an agent holding an item with non-positive marginal value and transfers that item as follows: If the agent with the non-positive marginal value item is the one with the least-valued bundle (aka agent $1$), the item is transferred to any bundle (excluding the special agent $i^*$) where it has strictly positive marginal value; 
otherwise, it is transferred to the least-valued bundle i.e., $A_1$, if it provides strictly positive marginal value there; if not, it is transferred to any agent (excluding the special agent $i^*$) for whom the item has strictly positive marginal value. This process continues until no such item remains, yielding a \NT{} allocation.

\begin{algorithm}[htp]
\begin{algorithmic}[1]
\small
\REQUIRE Special agent $i^*$ and a complete allocation $A=(A_1, \ldots, A_n)$ s.t. $v(A_1) \le \cdots \le v(A_n)$
\ENSURE  A \NT{} allocation s.t. $v(A_1) \le \cdots \le v(A_n)$ 
    \WHILE{(there exists an agent $i\in N$ with an item $o\in A_i$\ \text{s.t.}\ $v(A_i) \leq v(A_i\setminus\{o\})$ }\label{line:general_NT}
        \STATE If $i = 1$ (least valued bundle), then transfer $o$ to an agent $k\neq i^*$ s.t. $v(A_k \cup \{o\}) > v(A_k)$;
        \STATE Else if $i\neq 1$ and $v(A_{1}\cup \{o\}) > v(A_1)$, then transfer $o$ to $A_1$;
        \STATE Else transfer $o$ to any agent $k\neq i^*$ s.t. $v(A_{k}\cup \{o\}) > v(A_k)$.
         \STATE Relabel the bundles s.t. $v(A_1) \le \cdots \le v(A_n)$. 
    \ENDWHILE
\end{algorithmic}
\caption{Transfer Stability (\textsc{TS-Subroutine})} 
\label{alg:general_NT_subroutine}
\end{algorithm}

\paragraph{The Importance of $n\geq 4$ and Special Agent.}

The structure of cut-valuations imposes a limit on the number of agents who can view a given item as a chore (non-positive marginal value); as we show in \Cref{claim:chore_for_two_agents}, this can be at most two such agents.
This means that when $n\ge 3$, in a $\NT$ allocation, each item must be assigned to an agent who values it positively. 
As shown in \Cref{thm:general:non_exist}, sometimes this constraint conflicts with the \EFone{} requirement.

When $n\geq 4$, this bound on the number of agents guarantees the existence of \emph{multiple} agents for whom an item is positively valued.
This enables series of transfer to obtain \EFone{} when direct transfers to the least valued agent may not be possible: An item can be transferred to another agent who values it positively, and if needed, another item from their bundle can be passed to $A_1$.  
To ensure polynomial-time termination, we sometimes need to treat an agent as ``special'' and not alter its bundle during the \textsc{TS-Subroutine}.
In such cases, we again rely on the $n\ge 4$ condition to guarantee that there exists another agent---distinct from the special agent---who values an item positively. %

\begin{restatable}{lemma}{chorestwoagents}\label{claim:chore_for_two_agents} 
    Given a cut-valuation instance and a (partial) allocation, any item $o\in V$ has non-positive marginal value for at most two agents. 
    Moreover, when $n\geq 4$,  item $o$ has positive marginal value for at least two agents.
\end{restatable}

\begin{figure}[ht]
\centering
\scalebox{0.85}{
\begin{tikzpicture}[
    item/.style={draw, circle, fill=gray!20, minimum size=0.5cm, inner sep=1pt},
    neighbor/.style={draw, circle, fill=white, minimum size=0.4cm, inner sep=1pt},
    bundlebox/.style={draw, dashed, rounded corners},
    emptybox/.style={draw, dashed, rounded corners, thick, gray}, 
    font=\small
]


\node[item] (o1) at (0,0) {$o$};

\node[neighbor] (t1) at (-1.2,0.9) {};
\node[neighbor] (t2) at (-0.6,0.9) {};
\node[neighbor] (t3) at (0.0,0.9) {};
\node[neighbor] (t4) at (0.6,0.9) {};
\node[neighbor] (t5) at (1.2,0.9) {};

\draw[bundlebox, red!60!black] (-1.5,0.6) rectangle (1.5,1.2);
\node at (1.75,0.9) {$A_i$};

\node[neighbor] (b1) at (-1.0,-0.9) {};
\node[neighbor] (b2) at (-0.4,-0.9) {};
\node[neighbor] (b3) at (0.8,-0.9) {};  

\draw[bundlebox, blue!70!black] (-1.3,-1.17) rectangle (0.07,-0.63);
\draw[bundlebox, teal!70!black] (0.45,-1.17) rectangle (1.1,-0.63);
\node at (-1.65,-1.0) {$A_j$};
\node at (1.4,-1.0) {$A_k$};

\foreach \t in {t1,t2,t3,t4,t5} {
  \draw (o1) -- (\t);
}
\foreach \b in {b1,b2,b3} {
  \draw (o1) -- (\b);
}

\node at (0,-1.8) {\textbf{Item $o$ has negative marginal}};
\node at (0,-2.2) {\textbf{value for at most one agent.}};


\node[item] (o2) at (5.5,0) {$o$};

\node[neighbor] (t6) at (4.4,0.9) {};
\node[neighbor] (t7) at (5.0,0.9) {};
\node[neighbor] (t8) at (5.6,0.9) {};

\draw[bundlebox, red!60!black] (4,0.55) rectangle (6,1.2);
\node at (6.2,0.6) {$A_i$};

\node[neighbor] (b4) at (5.1,-0.9) {};
\node[neighbor] (b5) at (5.7,-0.9) {};
\node[neighbor] (b6) at (6.3,-0.9) {};

\draw[bundlebox, blue!70!black] (4.7,-1.17) rectangle (6.65,-0.61);
\node at (6.9,-1.0) {$A_j$};

\draw[emptybox, teal!70!black] (6.9,0.27) rectangle (7.7,1.13);
\node at (7.3,0.7) {$A_k$};

\foreach \x in {t6,t7,t8,b4,b5,b6} {
  \draw (o2) -- (\x);
}

\node at (5.5,-1.8) {\textbf{Item $o$ has zero marginal}};
\node at (5.5,-2.2) {\textbf{value for at most two agents}};
\end{tikzpicture}
}
\caption{An illustration of the characterization in \Cref{claim:chore_for_two_agents}. 
}
\label{fig:claim_chore_for_two_agents_illustration}
\end{figure}
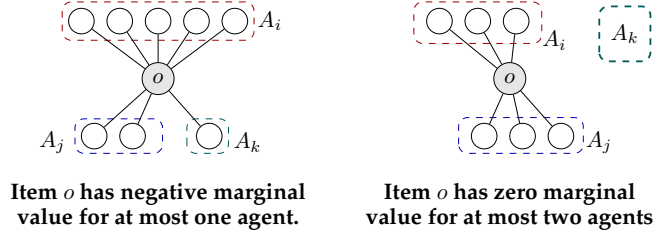

\begin{proof} 
Suppose $o$ has a negative marginal value for agent $i$. This means that $o$ has more neighbors inside $A_i$ than outside $A_i$, i.e., $|N_{A_i}(o)| > |N_{V\setminus A_i}(o)|$. For any other agent $j \neq i$, this implies that $N_{A_j}(o) \subseteq N_{V\setminus A_i} (o)$ and $N_{A_i}(o) \subseteq N_{V\setminus A_j}(o)$. Thus, $|N_{A_j}(o)| \le |N_{V\setminus A_i} (o)| < |N_{A_i}(o)| \le |N_{V\setminus A_j} (o)|$, i.e., $o$ has fewer neighbors inside $A_j$ than outside, so $o$ must have a positive marginal value (i.e. be a good) for all such $j$. For an illustration of this, see \Cref{fig:claim_chore_for_two_agents_illustration}.

Similarly, if agent $i$ has zero marginal value for some $o\in V$ (that is, $|N_{A_i}(o)| = |N_{V\setminus A_i}(o)|$ i.e. the number of neighbors of $o$ inside and outside $A_i$ is equal), then at most one other agent (say $j$) can have zero value for $o$ (when $N_{A_j}(o) = N_{V\setminus A_i}(o)$; see \Cref{fig:claim_chore_for_two_agents_illustration}), and all other agents must have strictly a positive marginal value for $o$. This proves the first part of the lemma, i.e. any item has a non-positive marginal value for at most two agents.
By the pigeonhole principle, when $n\ge 4$, then there are at least two agents for whom $o$ has a positive marginal value.
\end{proof}

Before proving \cref{thm:general:NT4}, we first present three key technical claims. 
First, we show that each agent derives at least half the total value of all bundles that contain only items with non-positive marginal value to that agent.
Intuitively, consider one such chore $o$ that is received by an agent $j$.
Since $o$ is a chore for $i$, there are at least as many edges from $o$ to $A_i$ as from $o$ to $V\setminus A_i$. 
These edges (from $o$ to $A_i$) contribute to the value of $A_i$. 
On the other hand, the value of $A_j$ can be at most equal to sum of all edges from $o\in A_j$, which is at most twice the number of edges from $o$ to $A_i$.

\begin{restatable}{claim}{TSallchores}
\label{ob:general:ne:all_chores}
    For an agent $i \in N$, consider any set of bundles $X$ such that for every bundle $A_k \in X$, it holds that $v(A_i) \ge v(A_i \cup \{o\})$ for every $o \in A_k$. Then, we have
    \[
        v(A_i) \ge \frac{1}{2}\sum_{A_k \in X}v(A_k). 
    \]
\end{restatable}

\begin{proof}
First, we note that when $o\in A_k$ is a chore with respect to $A_i$, we have
\begin{align}
    |\{(o, o^\prime)\in E:\, o^\prime \in A_i\}| \geq |\{(o, o^\prime)\in E:\, o^\prime \in V\setminus A_i\}|  \nonumber
\end{align}
Adding $|\{(o, o^\prime)\in E:\, o^\prime \in A_i\}|$ on both sides, we get
\begin{align}
    2|\{(o, o^\prime)\in E:\, o^\prime \in A_i&\}| \geq |\{(o, o^\prime)\in E :\, o^\prime \in A_i\}| \nonumber\\& + |\{(o, o^\prime)\in E :\, o^\prime \in V\setminus A_i\}| \nonumber
\end{align}
Dividing both sides by 2, we have 
\begin{equation}
    |\{(o, o^\prime)\in E :\, o^\prime \in A_i\}| \geq \frac{1}{2} |\{(o, o^\prime)\in E :\, o^\prime \in V\}| \label{eq:chore_property}
\end{equation}

Now, using \eqref{eq:chore_property}, we can lower bound the value of $A_i$ in the following manner.
\begin{align*}
    v(A_i) &= |\{(o, o^\prime)\in E:\, o \in A_i, o^\prime \in V \setminus A_i\}| \\
    &\ge \sum_{A_k \in X}|\{(o, o^\prime)\in E :\, o \in A_i, o^\prime \in A_k\}| \tag*{\text{(Since $A_k \subseteq V\setminus A_i$)}}\\
    &\ge \frac{1}{2} \sum_{A_k \in X} \sum_{o\in A_k} |\{(o, o^\prime)\in E:\, o^\prime \in V| \tag*{(From \eqref{eq:chore_property})}\\
    &\geq \frac{1}{2} \sum_{A_k \in X}v(A_k), \tag*{\text{(Since $v(S)\leq \sum_{o\in S} \deg(o)$)}}
\end{align*} thus proving the claim.
\end{proof}

Using this claim, it is easy to see that when Case II happens, i.e., when there is \EFone{} violation from agent $1$ to some agent $i^*$ but $A_{i^*}$ does not have any item with positive marginal value for $A_1$, such a violation can only be towards a \emph{single} agent.
Thus, all remaining bundles contain at least one item with positive marginal value.

\begin{restatable}{claim}{casetwospecialagentcondition}
\label{ob:general:ne:case2}
    {Let $X$ be the set of all bundles towards which the first agent (agent with least valued bundle) has \EFone{} violations. 
    If, for every bundle $A_k \in X$ and $o\in A_k$ we have $v(A_1) \ge v(A_1 \cup \{o\})$, then $|X| \le 1$. 
    Moreover, for any other bundle $A_j \notin X$, there exists an item $o^\prime \in A_j$ such that $v(A_1 \cup \{o^\prime\}) > v(A_1)$.}\footnote{Notice that the same proof also implies that if agent $1$ \textbf{envies} (not necessarily $\EFone$-envies) exactly one agent $i$ and all items in $A_i$ have non-positive value with respect to agent $1$, then every other agent has an item with a positive marginal value for agent $1$. This generalization helps us prove \Cref{thm:general:NT4}.}
\end{restatable}

\begin{proof}
    Let $i^*$ be an agent who is $\EFone$-envied by agent $1$. 
    From Case 2, we know that every item in $A_{i^*}$ must be a (weak or strict) chore with respect to $A_1$. 
    If there exists another agent (say $k$) whose items are all chores for agent $1$, by \Cref{ob:general:ne:all_chores}, we have that 
    \[
    v(A_1) \ge \frac{1}{2}\cdot (v(A_j) + v(A_k)), 
    \]
    which is a contradiction since agent $1$ strictly envies $i^*$ and $v(A_1) \le v(A_k)$. 
    Therefore, in Case 2 where all items in $A_{i^*}$ are chores with respect to $A_1$, it must be the case that $i^*$ is the only agent who is $\EFone$-envied by agent $1$.
\end{proof}

To argue about the termination of the main algorithm, \Cref{alg:general:NT}, we use a lexicographic improvement in the vector $\Phi = (v(A_1), - n_{1})$, where $v(A_1)$ represents the value of the least-valued bundle and $n_{1}$ represents the number of agents with this least-value.
Finally, to prove that the \textsc{TS-Subroutine} (\Cref{alg:general_NT_subroutine}) terminates in polynomial time, we again use a potential argument; in particular, we show that the social welfare strictly increases and $\Phi$ does not decrease.
Since $\SO \leq 2|E|$ and $O(|E|)\in O(m^2)$, the \textsc{TS-Subroutine} converges in $O(m^2)$ time. 

\begin{restatable}{claim}{TSsubroutineconvergence}\label{ob:NT-subroutine-converges}
    The \textsc{TS-Subroutine} returns an allocation that is \NT{}. Furthermore, in each iteration of \Cref{alg:general_NT_subroutine}, the potential function $\Phi = (v(A_1), -n_1)$ does not decrease, and the social welfare $\textsc{SW}(A) = \sum_{i\in N} v(A_i)$ increases by at least $1$. Therefore, \Cref{alg:general_NT_subroutine} terminates in $O(m^2)$ time. 
\end{restatable}

\begin{proof}
    First, we prove that if the subroutine terminates, then the returned allocation is \NT{}.
    Recall that in a \NT{} allocation, there must be no transfer that decreases the value of either of the agents. 
    Thus, we show that, in the returned allocation, the donator bundle's value reduces during any transfer, thereby proving that the allocation is \NT.
    Notice that, on termination, the while-loop condition ensures that every item is assigned to a bundle where its marginal value is positive. Thus, transferring any item will reduces the value of the donator bundle, proving that the allocation is \NT. 
    In the remaining part of the proof we argue that the while loop in the subroutine terminates in polynomial time.
    
    In each iteration of the \textsc{TS-Subroutine}, we identify an item that has been assigned as a strict/weak chore $o$ in some agent $i$'s bundle. 
    Then, $o$ is reallocated to another (carefully chosen) agent for whom $o$ is a strict good (postivie marginal value). 
    Such an agent always exists when $n\geq 3$ by \Cref{claim:chore_for_two_agents}.
    By the same lemma, when $n\geq 4$, there always exists an agent $k\neq i^*$ (not the special agent) for whom $o$ is a strict good (positive marginal value). 
    We transfer $o$ from a bundle where it is a weak or strict chore (zero or negative marginal value) to one (other than $A_{i^*}$) where it is a strict good (positive marginal value). 
    This guarantees that no agent's value decreases, and the receiving agent's value strictly increases. 
    Hence, $v(A_1)$ (the minimum value) does not decrease, $n_{1}$ (the number of agents that have bundles of value $v(A_1)$) does not increase, and the total social welfare increases by at least 1. 

    Since $\SO \leq 2|E|$ and $O(|E|) \in O(m^2)$, the social welfare can increase at most $\SO$ times, i.e. the \textsc{TS-Subroutine} must terminate in $O(m^2)$.
\end{proof}

We are now ready to prove \Cref{thm:general:NT4}.

\begin{proof}[Proof of Theorem \ref{thm:general:NT4}]
    First, we show that if \Cref{alg:general:NT} terminates, then the resulting allocation is both \EFone{} and \NT{}. After that, we prove that the algorithm terminates in polynomial time.
    By the main while loop in Line~\ref{line:EF1_while_NT}, the algorithm continues until agent 1 (the agent with the least value) no longer has an \EFone{} violation towards any other agent. By \Cref{lem:identical}, this ensures that the allocation is \EFone{}. Next, from \Cref{ob:NT-subroutine-converges}, we know that the $\textsc{TS-Subroutine}$ guarantees that the resulting allocation is \NT{}. Since this subroutine is called in Lines \ref{line:NT-subroutine-1} and \ref{line:NT-subroutine-2}, the final allocation also satisfies \NT{}. 
    
    We now prove that the algorithm terminates in polynomial time. We show that (i) in every iteration of Case I there is a lexicographic improvement in $\Phi = (v(A_1), -n_1)$; and (ii) in Case II, either (a) either $\Phi$ is lexicographically improved, or (b) when $\Phi$ does not change, the algorithm moves back to Case I. 
    Since each part of the algorithm runs in polynomial time---$v(A_1)$ can increase at most $O(|E|)\in O(m^2)$ times, $n_1$ can decrease at most $O(n)$ times, the while loop is Case II runs in $O(m)$ time, and the social welfare can improve at most $O(m^2)$ times %
    (which is the potential function for the \textsc{TS-Subroutine}; see \Cref{ob:NT-subroutine-converges})---the algorithm terminates in time polynomial in $n$ and $m$.

    \vspace{3pt}
    \noindent\textbf{Case I.} 
    Let $i$ be an agent for whom there is an \EFone{} violation and let $o\in A_i$ be an item with positive marginal value for $A_1$. 
    After reallocating $o$ to agent $1$, the value of $A_1$ strictly increases. 
    In addition, since there was a \EFone{} violation towards  $i$, we know $A_i$'s new value is still strictly larger than previous least value i.e. $v(A_i\setminus\{o\}) > v(A_1)$.  
    If $v(A_2) = v(A_1)$, the minimum value does not change, but the number of agents receiving the minimum value $n_1$ decreases by $1$. 
    Otherwise, the minimum value increases by at least $1$ and thus $\Phi$ lexicographically improves.
    
    \vspace{3pt}
    \noindent\textbf{Case II.} First, we show that in Case II, $\Phi$ does not decrease. 
    From \Cref{ob:general:ne:case2}, we know that any \EFone{} violation must be towards a \emph{single} special agent $i^*$ and all items in $A_{i^*}$ have a non-positive marginal value for $A_1$. 
    In this case, an arbitrary item $o\in A_{i^*}$ is transferred to an agent $k\ne i^*$ who has positive marginal value for $o$; the existence of such an agent $k\neq i^*$ is guaranteed via \Cref{claim:chore_for_two_agents} when $n\ge 4$.
    Note that $v(A_1)$ does not change since $A_1$'s bundle is untouched, and no agent becomes worse than $v(A_1)$: Agent $i^*$ is $\EFone$-envied by agent $1$ so $v(A_{i^*}\setminus \{o\}) > v(A_1)$, and $o$ has a positive marginal value for $A_k$.  

    Thus, for Case II, we need to show that if $\Phi$ does not lexicographically increase and the allocation is not \EFone{}, then we must move back to Case I, i.e. there cannot be two consecutive runs of Case II unless the potential function $\Phi$ increases.
    To see this, notice that after the execution of the subroutine in Line \ref{line:NT-subroutine-2}, if the resulting allocation is not \EFone{}+\NT{} and $v(A_1)$ and $n_1$ (i.e. $\Phi$) remain unchanged, then the least valued agent remains agent $1$. This is because, from \Cref{ob:NT-subroutine-converges}, we know that the subroutine does not make any agent worse-off. 
    Also, since $A_{i^*}$ was $\NT$ before Case II (i.e. contained no non-positive items for itself), it must be the case that $i^*$ was not involved in any transfers during the subroutine execution. Thus, any $\EFone$-violation must be towards a different agent since agent 1 did not \EFone{}-envy $i^*$ after Case II. 
    Moreover, since the least valued agent is still agent $1$ and (the new) $A_{i^*}$ is still an envied  bundle with no items that have a positive marginal value for $A_1$, from \Cref{ob:general:ne:case2}, the new \EFone{}-envied agent must have an item with positive marginal value for $A_1$, satisfying the condition for Case I.  
\end{proof}

\subsection{Fairness and Efficiency: A Monotonic Result
}\label{sec:arbitrary_n_positive_results}

We have shown that \EFone{} and the efficiency notion $\textsc{TS}$ exhibits non-monotonic existence.
This raises a natural question: \textit{Is there an efficiency notion that is always compatible with \EFone{}, regardless of $n$?}
We answer this question positively by showing that a slight relaxation of $\textsc{TS}$--namely \wNT{}--is guaranteed to exist along with \EFone{} for any $n$ (\Cref{thm:general:ne}).
Towards this, we provide a polynomial-time algorithm that computes such allocations.
We also show that, if we are willing to relax fairness in favor of a stronger efficiency notion, the leximin solution guarantees $\frac{1}{2}$-\EFone{} and is \PO{} for any $n$ (\Cref{thm:approximate_EF1_PO}). We emphasize that this is only an existence guarantee, as computing an $\alpha$-\EFone{} allocation that is also \PO{} is \textsc{NP-hard} for any $\alpha$.

\subsubsection{Weakening Efficiency}
In this section, we show that \EFone{} is compatible with \wNT{}, a slightly weaker notion than \NT{}, for instances with general graphs and any number of agents.
\begin{restatable}[\EFone{}+\wNT{} allocations]{theorem}{EFwTS}
\label{thm:general:ne}
    Given a cut-valuation instance, a complete allocation satisfying \EFone{} and \wNT{} always exists and can be computed in polynomial time.
\end{restatable}

Our proof is constructive: we provide an algorithm, akin to \Cref{alg:general:NT}, that computes such an allocation for any instance. The full algorithm and proof of \Cref{thm:general:ne} can be found in \Cref{sec:proofs_sec_general}. 
Intuitively, unlike \NT{} which requires each item to have positive marginal in its bundle, \wNT{} permits items to have zero-marginals. This forgoes the need for multiple agents with positive value for an item, removing the need for $n\ge 4$ (\Cref{claim:chore_for_two_agents}).

As a corollary, \Cref{thm:general:ne} implies that any graph with $m$ vertices admits a non-empty $n$-partition of the vertex set (with $n\leq m$) with nearly-equitable the cut-values across parts. 
Notably, the bound of $\Delta$ in \Cref{cor:equitable_cuts} improves on the previous best of $5\Delta + 1$ by \citet{barman2025fair}. The proof of the Corollary is presented in \Cref{sec:proofs_sec_general}.

\begin{restatable}[Equitable graph cuts]{corollary}{equitablegraphcuts}\label{cor:equitable_cuts}
    Given a graph $G = (V, E)$ and an integer $n\leq |V|$, there exists a polynomial-time computation $n$-partition of $V$ into non-empty $V_1, V_2, \ldots , V_n \neq \emptyset$ such that the cut-values satisfy $|v(V_i) - v(V_j)| \leq \Delta$ for all $i, j \in [n]$, where $\Delta:= \max_{o\in V} \deg(o)$ denotes the maximum degree of $G = (V,E)$.
\end{restatable}

\subsubsection{Weakening Fairness} \label{app:appEF1PO}

Recall, \Cref{thm:general:non_exist} shows that, for general graphs, an allocation that is both $\EFone$ and $\PO$ need not always exist. 
In this section, we show that a constant multiplicative approximation of $\EFone$ alongside $\PO$ always exists, where a multiplicative approximation of $\EFone$ is defined as follows.

\paragraph{$\alpha$-$\EFone$.} An allocation $A = (A_1, \ldots, A_n)$ is said to be $\alpha$-$\EFone$ if for every $i,j\in N$ such that $v(A_j) > v(A_i)$, there exists an item $o_j\in A_j$ such that $v(A_i) \geq \alpha\ v(A_j \setminus \{o_j\})$ where $\alpha \in (0,1]$.

In particular, we show that the leximin solution (defined below) is $\frac{1}{2}$-$\EFone$.

\paragraph{Leximin} An allocation $A=(A_1, \ldots, A_n)$ is said to be leximin optimal if $A$ maximizes the value of the agent with the smallest-value, subject to that maximizes the value of the second smallest agent's bundle, and so on.

\begin{proposition}[Approximate $\EFone$ and PO]\label{thm:approximate_EF1_PO}
    Given a cut-valuation instance, there always exists an allocation that is $\frac{1}{2}\text{-}\EFone$ and $\PO$.
\end{proposition}

Note that computing the leximin solution is $\textsc{NP-Hard}$ even for $n=2$ (same as \textsc{max-cut}). 
Thus, \Cref{thm:approximate_EF1_PO} is only an existence result; however, efficient computation of $\alpha$-$\EFone$ and $\PO$ allocations is a futile endeavor for any $\alpha$ since finding a $\PO$ allocation is NP-hard, even for $n=2$ (\Cref{prop:n=2}).

\vspace{0.3cm}

\begin{proof}[Proof of \Cref{thm:approximate_EF1_PO}]
    We show that a leximin solution is $\frac{1}{2}$-$\EFone$. 
    Since leximin solutions are $\PO$, this completes the proof. 
    For a subset of items $S\subseteq V$ and an item $o \in V$, recall that $N_S(o)$ denotes the number of neighbors of $o$ that belong to $S$. 
    Thus, the value of any set $S$, i.e., the number of edges with one endpoint in $S$ and the other in $V\setminus S$, is given by $v(S) = \sum_{o\in S} |N_{V\setminus S} (o)|$.
    
    Let $A=(A_1, \ldots, A_n)$ be a leximin allocation, and let $i,j$ be two agents such that $i$ $\EFone$-envies $j$, i.e., for every $o\in A_j$, we have $v(A_j\backslash\{o\})> v(A_i)$. 
    Since $A$ is leximin, addition of any item $o\in A_j$ to $A_i$ does not increase the value of $A_i$, that is, every $o\in A_j$ must have a non-positive marginal value for $A_i$ i.e. $N_{A_i}(o) \geq N_{V\setminus A_i}(o)$. 
    Thus, an upper-bound on the value of $A_j$ can be obtained as
    \begin{align*}
        &v(A_j) = \sum_{o\in A_j} |N_{V\setminus A_j}(o)| \\
        &= \sum_{o\in A_j} \bigg[|N_{V\setminus A_j}(o)\cap N_{V\setminus A_i}(o)| + |N_{V\setminus A_j}(o)\cap N_{A_i}(o)|\bigg]\\
        &\leq  \sum_{o\in A_j} \bigg[|N_{V\setminus A_j}(o)\cap N_{A_i}(o)|+ |N_{V\setminus A_j}(o)\cap N_{A_i}(o)|\bigg] \\
        &= 2\sum_{o\in A_j} |N_{V\setminus A_j}(o)\cap N_{A_i}(o)|\\
        &\leq 2\ v(A_i). 
    \end{align*}
\end{proof}

\section{Forest Graphs}\label{sec:forests}

In this section, we provide a polynomial-time algorithm that takes as input a cut-valuation instance where the graph $G = (V, E)$ is a forest (i.e. acyclic), and outputs an allocation that is simultaneously \EFone{} and socially optimal (\SO). This algorithmic guarantee also shows the \emph{existence} of such allocations for forest graphs.
Our first observation is a characterization of \SO{} allocations for forests: When $n \geq 2$, no two adjacent items can be allocated to the same agent.

\begin{restatable}{proposition}{forestsSO}\label{prop:SO_forests}
    Given a cut-valuation instance where $G$ is a forest and $n\ge 2$, an allocation $A$ is \SO{} if and only if, for every $(o, o')\in E$, we have $o\in A_i$ and $o'\in A_j\neq A_i$.
\end{restatable}
\begin{proof}
    A forest is a graph that is acyclic. Since a graph is bipartite if and only if it contains no cycles of odd length, a forest must be bipartite. Let the bipartite graph corresponding to the forest be $G = (V_1 \sqcup V_2, E)$. Consider the  social welfare of the allocation $A = (V_1, V_2, \emptyset, \ldots, \emptyset)$. 
    \[
    \textsc{SW}(A) = \sum_{i \in N} v(A_i) = v(V_1) + v(V_2) = 2|E|.
    \]
    In any graph, since each edge can at most contribute to the valuation of both its endpoints, the optimal social welfare is bounded above by $2|E|$, i.e., $\SO \leq 2|E|$. 
    Thus, the optimal welfare for forest graphs is exactly $2|E|$, which means that no adjacent items can be allocated to the same agent in any $\SO$ allocation.
\end{proof}
As consequence of \Cref{prop:SO_forests}, the problem of finding an allocation that is \EFone{} and \SO{} is equivalent to finding an $\EFone$ allocation where no adjacent vertices (items) are assigned to the same agent. Under this condition, agent's bundle forms an independent set $S$ in the graph, with $v(S) = \sum_{o\in S} \deg(o)$.
This coincides with the conflict-constraints, albeit for valuations derived from cut-functions. 

Note that for $n = 2$, the bi-partition based on the bipartite graph of the forest is both \EF{} and \SO{} since the cut function is symmetric, i.e. $v(S) = v(V\setminus S)$ for any $S\subseteq V$\footnote{In fact, this is true in more generality for \emph{any} bipartite graph (not just forests) and $n=2$ agents.}; thus we assume $n\geq 3$. 
Interestingly, the negative example in \Cref{thm:general:non_exist} is also a bipartite graph (with treewidth = 2), but an \EFone{}+\SO{} allocation does not exist for it when $n=3$. This highlights that while strong existence and algorithmic guarantees (\EFone{}+\SO{}) hold for forests (bipartite; treewidth 1), they break even for graphs that are only slightly different.

\paragraph{Algorithm Overview.}
\cref{alg:EF1_SO_trees} begins with an empty initial allocation and an arbitrary rooting of each tree in the forest. 
The algorithms proceeds iteratively: at each step, it assigns a root item (vertex) to an agent and roots the resulting subtrees at the unallocated neighbors (children) of the assigned item. 
Since we always exclusively assign root items, at most one of their neighbors---the parent---may have been already assigned. 
Therefore, to satisfy the \SO{} condition, it suffices to ensure that the assigned root item is not in the same bundle as its parent. 
We refer to such a root as a \emph{feasible root} (for a bundle)---i.e. one whose parent does not belong to that bundle.  

Throughout the algorithm, we relabel bundles to ensure that $v(A_1) \le v(A_2) \le \ldots \le v(A_n)$ holds at all times (\Cref{lem:identical}). 
Whenever there exists a feasible root for the least valued agent, aka agent $1$, the algorithm assigns this root to $A_1$ as part of 
\textsc{Case 1}.
\textsc{Cases 2} and \textsc{3} handle situations where no root item is feasible for $A_1$. 
Here, the algorithm assigns a root item to a different agent, temporarily violating the \EFone{} condition. 
However, it is carefully designed to \emph{catch up} on such \EFone{} violation through subsequent assignments by assigning $o_t$'s children to agent 1 (or others as needed).
While doing so, we first prioritize the \textbf{leaf-children} of a node $o_t$, i.e. the children of $o_t$ (unallocated neighbors) that have a degree of $1$. 
This helps us ensure that, as long as there is an unallocated item, at least one of the three \textsc{Cases} must hold.

\begin{algorithm}[H]
\begin{algorithmic}[1]
\small
\STATE Initialize with an empty allocation $A=(\emptyset, \ldots, \emptyset)$. 
\STATE Root each of the $k$ trees at arbitrary nodes $r_1, \ldots, r_k$.
\WHILE{(there exists an unallocated item)}
\STATE Relabel bundles s.t. $v(A_1) \leq v(A_2) \leq \ldots \leq v(A_n)$.

\vspace{10pt}

\fcolorbox{gray}{white}{%
  \parbox{\dimexpr\linewidth-5\fboxsep-5\fboxrule}{%
    \textbf{\textsc{Case} 1:} The minimum-valued agent has a feasible root.
  }%
}
\vspace{2pt}

\IF{(agent $1$ has a feasible root node $o_t$)}
    \STATE Allocate $o_t$ to agent $1$, i.e., $A_1\gets A_1 \cup \{o_t\}$.
    \WHILE{(there is an unallocated leaf-child $o_\ell$ of $o_t$)} \label{line:case-1-leaf-child}
    \STATE Let $F = \argmin_{i\in[2,n]} v(A_i)$, and consider an agent $j\in F$. If $|F|\geq 2$, choose $j\in F$ such that there exists an agent in $F\setminus\{j\}$ with a feasible root.\label{line:tie-break_case1}
    \STATE Allocate $o_\ell$ to $j$, i.e., $A_j\gets A_j \cup \{o_\ell\}$.
    \ENDWHILE \label{line:case-1-leaf-child-end}

    \vspace{12pt}
    
\hspace*{-10pt}\fcolorbox{gray}{white}{%
  \parbox{\dimexpr\linewidth-2\fboxsep-2\fboxrule}{%
    \textbf{\textsc{Case} 2:} A feasible root can be allocated to the second minimum-valued agent.
  }%
}

\ELSIF{($v(A_1) > v(A_2\setminus \{o_2^*\})$ or agent $2$ has a feasible root $o_t$ such that $v(o_t)>v(o_2^*)$)}
    \STATE Allocate $o_t$ to agent $2$, i.e., $A_2\gets A_2 \cup \{o_t\}$. 
    \STATE Let $h_2^*\in A_2$ be s.t. $h_2^{*} \in \argmin_{h\in A_2} v(A_2\setminus\{h\})$.
    \WHILE{(there is an unallocated leaf-child $o_\ell$ of $o_t$)} \label{line:2while_start}
    \STATE Let $F = \argmin_{i\in[1]\cup[3,n]} v(A_i)$, and consider an agent $j\in F$. If $|F|\geq 2$, choose $j\in F$ such that there exists an agent in $F\setminus\{j\}$ with a feasible root. \label{line:tie-break_case2}
    \STATE Allocate $o_\ell$ to $j$, i.e., $A_j\gets A_j \cup \{o_\ell\}$.\label{line:case2_leaf_children}
    \ENDWHILE\label{line:2while_end}
    \WHILE{($v(A_1) < v(A_2\setminus\{h_2^*\})$)}\label{line:case2:non-leaf:begin}
        \STATE Allocate an unallocated child of $o_t$ to agent $1$.\label{line:case2_more_children}
    \ENDWHILE\label{line:case2:non-leaf:end}
    \vspace{12pt}

\hspace*{-10pt}\fcolorbox{gray}{white}{%
  \parbox{\dimexpr\linewidth-2\fboxsep-2\fboxrule}{%
    \textbf{\textsc{Case} 3:} A feasible root can be allocated to some other agent.
  }%
}
\vspace{3pt}

\ELSIF{(there exists $j\in \argmin_{i\in [3,n]} v(A_i\setminus \{o_i^*\})$ such that $v(A_1) > v(A_j \setminus\{o_j^*\})$)}
    \STATE Allocate $o_t$ to agent $j$, i.e., $A_j\gets A_j\cup\{o_t\}$.
    \STATE Allocate all leaves of $o_t$ to agent $1$. \label{line:case-3-leaf-child}
    \WHILE{($v(A_1) <\min \{v(A_2), v(A_j\setminus\{o_j^*\})\}$)} \label{line:case3_allocate children condn}
        \STATE Allocate an unallocated child of $o_t$ to agent $1$. \label{line:case3_allocate children condn:in} 
    \ENDWHILE \label{line:case3_allocate children condn:end}
\ENDIF
\ENDWHILE
\end{algorithmic}
\caption{Rooting peeling algorithm for $\EFone$ + $\SO$ allocations on forests 
} \label{alg:EF1_SO_trees}
\end{algorithm}

In particular, the algorithm proceeds by checking for the following three conditions sequentially:

\vspace{-3pt}

\begin{itemize}
    \item \textsc{Case 1:} \emph{Is there a feasible root $o_t$ for the least-valued agent (agent $1$)?} 
    \item \textsc{Case 2:} \emph{Can we allocate a feasible root $o_t$ to the second least-valued agent (agent $2$) such that any resulting $\EFone$ violation can be compensated for?} 
    \item \textsc{Case 3:} \emph{Is there an agent $j\neq 1,2$ who can receive a feasible root $o_t$ such that any resulting $\EFone$ violation can be remedied?} 
\end{itemize}

Throughout, we use $o_i^*$ to denote the item whose removal from $A_i$ causes the maximum drop in value of $A_i$, i.e., $o_i^* \in \argmin_{o\in A_i} v(A_i\backslash\{o\})$.

\begin{theorem}
[Existence of $\EFone$ + $\SO$ allocations]
\label{thm:forest}
    Given cut-valuation instance where $G = (V,E)$ is a forest, a complete allocation that is \EFone{} and \SO{} always exists and can be computed in polynomial time.
\end{theorem}

The algorithm, by construction, always assigns feasible root items, thereby ensuring that no two adjacent items are allocated to the same agent, i.e., the returned allocation is $\SO$. 
For \EFone{}, assigning an item to agent $1$ in \textsc{Case 1} preserves \EFone{}, since on removal of the last added item, $A_1$ was the least-valued bundle.
In contrast, when an item $o_t$ is assigned to an agent $j\neq 1$, each child of $o_t$ (unallocated neighbor of $o_t$) becomes a feasible root item for agent $1$, where each child has a value of at least $1$. 
Note that $o_t$ has at least $v(\{o_t\})-1$ (unassigned) children (since at most one neighbor of $o_t$, i.e. its parent has been assigned so far). 
Thus, agent $1$ can be \emph{almost} compensated for the resulting $\EFone$-violation (of at most $v(\{o_t\})$ toward $j$! 
\textsc{Cases 2} and 3 capture cases where this $\EFone$-violation can be \emph{fully} eliminated. 
In particular, \Cref{lem:EF1_every_loop_trees} shows that, after every execution of each of the \textsc{Cases}, the \EFone{} property is maintained.

\begin{restatable}{lemma}{lemtreesEFone}\label{lem:EF1_every_loop_trees}
    Consider an iteration of \Cref{alg:EF1_SO_trees} in which \textsc{Case X} is executed (where \textsc{X}$\in [3]$), and let the (partial) allocation before the iteration be $A^\mathrm{old}$ and after the iteration be $A^\mathrm{new}$. If $A^\mathrm{old}$ is $\EFone$, then $A^\mathrm{new}$ is also $\EFone$. 
\end{restatable}

Since the algorithm starts with an empty allocation which is trivially \EFone{}, inductively one can argue that the final resulting allocation will also be \EFone{}.
It remains to show that at least one of the three \textsc{Cases} is always satisfied whenever an item remains to be allocated [\Cref{lem:only_these_three_cases}]. 
In \Cref{lem:only_these_three_cases}, we show that \Cref{alg:EF1_SO_trees} always terminates with a complete allocation; that is, as long as there is an unallocated item, it can never be the case that none of the if-conditions ($\textsc{Cases}$) is satisfied.

\begin{restatable}{lemma}{completeallocation}
    During the execution of the algorithm whenever there is an unallocated item, at least one of the three \textsc{Cases} must hold. \label{lem:only_these_three_cases}
\end{restatable}

Combining these lemmas, we get that \Cref{alg:EF1_SO_trees} returns a complete allocation (by \Cref{lem:only_these_three_cases}) in $O(m)$ time that is \EFone{} (by \Cref{lem:EF1_every_loop_trees}) and \SO{} (by construction). Thus, \Cref{thm:forest} stands proved. 
The proofs of these lemmas are provided in \Cref{sec:proofs_forests}. 

\subsection{Proof of \Cref{thm:forest}}\label{sec:proofs_forests}

As discussed, our algorithm, by construction ensures that the returned allocation is \SO{}---no adjacent vertices (items) are allocated to the same agent. 
Thus, in order to prove \Cref{thm:forest}, we need to show that the \Cref{alg:EF1_SO_trees} returns a complete allocation (\Cref{lem:only_these_three_cases}) that is also \EFone{} (\Cref{lem:EF1_every_loop_trees}).
Before we prove these lemmas, we make a useful observation: At any iteration of the algorithm, no unallocated root node is leaf since all leaf-children are assigned in the same iteration as its parent. This is by construction---Lines \ref{line:case-1-leaf-child} to \ref{line:case-1-leaf-child-end}, \ref{line:2while_start} to \ref{line:2while_end}, and \ref{line:case-3-leaf-child} in \Cref{alg:EF1_SO_trees} ensure that all leaf-children of a root node $o_t$ are assigned in the same iteration as $o_t$.

\begin{observation}
     During any iteration of the algorithm, no unallocated root node is a leaf.  \label{lem:o_t not leaf}
\end{observation}

Throughout this section, we assume that $n\ge 3$ and that there are no isolated nodes (i.e. all nodes have degree at least $1$).
With this, we are now ready to prove the two lemmas. We begin by proving that the returned allocation is \EFone{}. Toward this, we prove a stronger claim --- in any iteration of the algorithm, if the (partial) allocation before the iteration was \EFone{}, then the (partial) allocation after this iteration will also be \EFone{}.
Since the algorithm starts with an empty allocation that is trivially \EFone{}, inductively one can argue that the final resulting allocation will also be \EFone{}.

\lemtreesEFone*

\begin{proof}
Since the algorithm relabels the bundles to maintain non-decreasing bundle valued from $A_1$ to $A_n$, we have
\begin{align}
    v(A_1^\mathrm{old})\leq v(A_2^\mathrm{old})\leq \ldots \leq v(A_n^\mathrm{old}) \label{eq:order_of_bundles}
\end{align}
Furthermore, let $o_i^*\in\argmin_{o\in A_i^\mathrm{old}}v(A_i^\mathrm{old}\setminus\{o\})$ and similarly $h_i^*\in\argmin_{h\in A_i^\mathrm{new}}v(A_i^\mathrm{new}\setminus\{h\})$. 

Our analysis is based on which case $X\in[3]$ is executed. For each of these, we show that there is no \EFone{} violation between any pair of agents. We start by considering that \textsc{Case 1} was executed in going from $A^\textrm{old}$ to $A^\textrm{new}$. 

\vspace{0.5em}
    \noindent\textbf{\underline{\textsc{Case 1}:}} The assignment in \textsc{Case 1} of \Cref{alg:EF1_SO_trees} is such that agent $1$ receives $o_t$ and other agents receive the leaf-children of $o_t$, if any. That is, no agent's bundle value goes down. Consider $i, j \in [2, n]$.
    
    \paragraph{\boldmath From $i\to 1\ $:} First, we show that there is no $\EFone$-envy from an agent $i\in [2,n]$ towards agent $1$. Since in the value of $i$'s bundle does not decrease ($i$ may or may not receive an additional leaf-child of $o_t$), we have
    $
        v(A_i^\mathrm{new})\geq v(A_i^\mathrm{old}) \geq v(A_1^\mathrm{old}) = v(A_1^\mathrm{new}\setminus \{o_t\}).
    $ The second inequality follows from \eqref{eq:order_of_bundles}, and the third inequality is due to the update $A_1^\mathrm{new} = A_1^\mathrm{old}\cup \{o_t\}$. Thus, no agent has an \EFone{} violation towards agent $1$.
    
    \paragraph{\boldmath From agent $1\to i\ $:} Next, we show that there is no \EFone{} violation from agent $1$ to any other agent $i\in [2,n]$. Since $A^\mathrm{old}$ was $\EFone$, we know that $v(A_1^\mathrm{old})\geq v(A_i^\mathrm{old}\setminus\{o_i^*\}).$ 
    Also, since a node cannot have more children than its own value (which is equal to its degree), we have  
    $
        v(o_t) = \deg(o_t) \ge  \#\ \text{leaf-children of }o_t.
    $
    Thus, even if $i$ receives \emph{all} of $o_t$'s leaf-children, we get,
    \begin{align*}
        v(A_1^\mathrm{new}) &= v(A_1^\mathrm{old})+v(o_t)
        \ge v(A_i^\mathrm{old}\setminus \{o_i^*\}) + \#\ \text{leaf-children of }o_t
        \ge v(A_i^\mathrm{new}\setminus\{o_i^*\}).
    \end{align*}
    That is, agent $1$ does not $\EFone$-envy any agent $i\in[2,n]$. 

    \paragraph{\boldmath From $i\to j\ $:} Finally, consider any $i, j\in[2,n]$ such that $v(A_j^\mathrm{new})>v(A_i^\mathrm{new})$. 
    If $j$ does not receive any item in this iteration, since $A^\mathrm{old}$ was $\EFone$, we get
    \[v(A_i^\mathrm{new})\geq v(A_i^\mathrm{old})\geq v(A_j^\mathrm{old}\setminus\{o_j^*\}) = v(A_j^\mathrm{new}\setminus\{o_j^*\}),\] implying that $i$ does not $\EFone$-envy $j$.
    On the other hand, if $j$ receives an item in this iteration, then by Line \ref{line:tie-break_case1} in \Cref{alg:EF1_SO_trees}, we have that $v(A_j^\mathrm{new}) =  v(A_i^\mathrm{new})+1$ since only leaf-children, whose value is $1$, are allocated to the least-bundle in $[2,n]$. Thus, removal of any item from $A_j$ eliminates this envy since the graph contains no isolated vertices with value $0$.


 
    \vspace{0.5cm}

    \noindent \textbf{\underline{\textsc{Case 2}:}} The assignment is \textsc{Case 2} of is such that agent $2$ receives a root node $o_t$, and the children of $o_t$, starting with the leaf-children, are assigned to the agent in ${1}\cup [3,n]$ with the least value. Consider $i, j\in [3,n]$.
    
    \paragraph{\boldmath From $i\to 2\ $:} For any agent $i\in [3,n]$ to agent $2$, notice from \eqref{eq:order_of_bundles}, that $v(A_2^\mathrm{old})\leq v(A_i^\mathrm{old})$. 
    Thus, we have
    $v(A_i^\mathrm{new}) \geq v(A_i^\mathrm{old})\geq v(A_2^\mathrm{old}) = v(A_2^\mathrm{new}\setminus \{o_t\}),$ where
    the final inequality is because of the update $A_2^\mathrm{new} = A_2^\mathrm{old}\cup\{o_t\}$. 
    Thus, no agent $i\in [3,n]$ $\EFone$-envies agent $2$ in $A^\mathrm{new}$.

    \paragraph{\boldmath From agent $2\to i\ $:} Notice that, in $A^\mathrm{new}$, agent $2$ (who receives $o_t$) does not $\EFone$-envy any other agent $i\in[3,n]$ (who may receive some leaf-children of $o_t$). 
    This is because $A^\mathrm{old}$ was $\EFone$, and $v(o_t)\geq  \#\ \text{leaf-children of }o_t$, whereby
    \begin{align*}
        v(A_2^\mathrm{new}) = v(A_2^\mathrm{old})+v(o_t)\geq  v(A_i^\mathrm{old}\setminus \{o_i^{*}\})  + \#\ \text{leaf-children of }o_t\geq v(A_i^\mathrm{new}\setminus\{o_i^*\}). 
    \end{align*}

    \paragraph{\boldmath From $i\to j\ $:} If $v(A_j^\mathrm{new})>v(A_i^\mathrm{new})$ and $j$ does not receive any item in this iteration (i.e. $A^\textrm{new} = A^\textrm{old}$), then
    $v(A_i^\mathrm{new})\geq v(A_i^\mathrm{old})\geq v(A_j^\mathrm{old}\setminus\{o_j^*\}) = v(A_j^\mathrm{new}\setminus\{o_j^*\}).$ That is, there is no \EFone{}-envy from $i$ to $j$. 
    If $j$ receives an item in this iteration, then by Line \ref{line:tie-break_case2} in \Cref{alg:EF1_SO_trees}, we have that $v(A_j^\mathrm{new}) =  v(A_i^\mathrm{new})+1$ (value of a leaf is $1$), implying that the removal of any item from $A_j$ eliminates this envy. 

    \paragraph{\boldmath From agent $1 \to i\ $ and $1\to 2$:} If some agent $i\in[3,n]$ also receives an item, we know that $v(A_1^\mathrm{new})\geq v(A_i^\mathrm{old})\geq v(A_2^\mathrm{old}) = v(A_2^\mathrm{new}\setminus \{o_t\})$. 
    Towards an agent $j\in [3,n]$ if $j$ does not receive an item, we know that $v(A_1^\mathrm{new})>v(A_1^\mathrm{old})\geq v(A_j^\mathrm{old}\setminus\{o_j^*\}) =  v(A_j^\mathrm{new}\setminus\{o_j^*\})$. 
    If $j$ receives an item, then by Line \ref{line:tie-break_case2}, we know $v(A_j^\mathrm{new})\leq v(A_1^\mathrm{new})+1$; thus, on removal of any item from $j$'s bundle, any envy will be eliminated.

    If no agent in $i\in [3,n]$ receives an item: Agent $1$ does not $\EFone$-envy $i$ since $A^\mathrm{old}$ is $\EFone$, and $v(A_1^\mathrm{new})\geq v(A_1^\mathrm{old})$. 
    Towards agent $2$, by the while condition in Line \ref{line:case2:non-leaf:begin}, we know that when the while loop terminates, there is no $\EFone$-envy from agent $1$ to $2$. 
    It remains to show that the while loop will terminate, i.e., there are enough children of $o_t$ to allocate to agent $1$ such that $v(A_1^\mathrm{new})$ can reach $v(A_2^\mathrm{new} \setminus \{h_2^*\})$.
    Since $o_t$ has $v(o_t)-1$ children, each of which contributes at least $1$ to agent $1$'s bundle when allocated, giving $o_t$'s children to agent $1$ can increase her value by at least $v(o_t)-1$. 
    If $v(A_1^\mathrm{old})\geq v(A_2^\mathrm{old})-v(o_2^*)+1$, we have 
    \begin{align*}
        v(A_1^\mathrm{old}) + v(o_t)-1 &\ge v(A_2^\mathrm{old})-v(o_2^*) + v(o_t)\\
        &= v(A_2^\mathrm{new} \setminus v(o_2^*))\ge v(A_2^\mathrm{new} \setminus v(h_2^*)). 
    \end{align*}
    On the other hand, if $v(o_t) \geq v(o_2^*)+1$, we know that $h_2^* = o_t$. 
    Then, 
    \begin{align*}
        v(A_1^\mathrm{old}) + v(o_t)-1 &\ge v(A_2^\mathrm{old})-v(o_2^*) + v(o_t) - 1\\ &\ge v(A_2^\mathrm{old}) = v(A_2^\mathrm{new} \setminus v(h_2^*)). 
    \end{align*}

    \paragraph{\boldmath From $i \to 1\ $ and $2\to 1$:} With this, we are only left to evaluate the possible envy from agent $i\in [2,n]$ to agent $1$.
    For this, we perform the analysis by considering two cases based on whether the last allocated item of this iteration ($o_f$) was done in Line \ref{line:case2_leaf_children} or in Line \ref{line:case2_more_children}.

    \emph{If the last allocated item was in Line \ref{line:case2_leaf_children}:} If $v(A_1^\mathrm{new})>v(A_i^\mathrm{new})$, then from Line \ref{line:tie-break_case2}, it must be the case that $v(A_1^\mathrm{new}) = v(A_i^\mathrm{new})+1$, whereby any envy can be eliminated by removing any item from agent $1$'s bundle. 

    \emph{If the last allocated item was in Line \ref{line:case2_more_children}:} From the while condition in Line \ref{line:case2_more_children}, we know that the $v(A_1^\mathrm{new}\setminus\{o_f\}) < v(A_2^\mathrm{new}\setminus\{h_2^*\})$. 
    Furthermore, we know from the allocation process that $A_2^\mathrm{new}=A_2^\mathrm{old}\cup \{o_t\}$ whereby $ v(A_2^\mathrm{new}\setminus\{o_t\})=v(A_2^\mathrm{old})$. 
    Finally, from \eqref{eq:order_of_bundles}, we know that $v(A_2^\mathrm{old})\leq v(A_i^\mathrm{old})\leq v(A_i^\mathrm{new})$; the second inequality is because no agent's bundle decreases in value. 
    Thus, combining all, we have
    \begin{align*}
        v(A_1^\mathrm{new}\setminus\{o_f\}) < v(A_2^\mathrm{new}\setminus\{h_2^*\}) \leq v(A_2^\mathrm{new}\setminus\{o_t\})=v(A_2^\mathrm{old})\leq v(A_i^\mathrm{old})\leq v(A_i^\mathrm{new}).
    \end{align*}
    That is, no agent $i\in[2,n]$ $\EFone$-envies agent $1$.

\vspace{0.5cm}

With this, we consider the final case and show that if \textsc{Case 3} was executed while going from $A^\mathrm{old}\to A^\mathrm{new}$ and if $A^\mathrm{old}$ was \EFone{}, then $A^\mathrm{new}$ must also be \EFone{}.
    
    \vspace{0.5cm}

    \noindent \textbf{\underline{\textsc{Case 3}:}} Let $j\in [3,n]$ be the agent who receives the feasible root $o_t$. Also, consider an arbitrary agent $i\in[2,n]$ such that $i\neq j$. Notice that any such agent $i$ is not allocated an item in this iteration.
    
    \paragraph{\boldmath From any agent $k\ne i$ to $i$:}
    For any $i$, notice that $v(A_i^\mathrm{new}) = v(A_i^\mathrm{old})$. 
    Also, for every agent $k$, we have, $v(A_k^\mathrm{new})\geq v(A_k^\mathrm{old})$. 
    Thus, since $A^\mathrm{old}$ was $\EFone$, we get that 
    \[v(A_k^\mathrm{new})\geq v(A_k^\mathrm{old})\geq v(A_i^\mathrm{old}\setminus \{o_i^*\})= v(A_i^\mathrm{new}\setminus \{o_i^*\})\]

    \paragraph{\boldmath From $i\to j\ $:}
    Since \textsc{Case 2} does not hold but $A^\mathrm{old}$ is $\EFone$, we get that $v(A_1^\mathrm{old}) = v(A_2^\mathrm{old}\setminus \{o_2^*\})$. 
    Thus, from the condition of \textsc{Case 3}, we get
    $
        v(A_j^\mathrm{old}) - v(o_j^*) < v(A_1^\mathrm{old}) = v(A_2^\mathrm{old}) - v(o_2^*).
    $ Adding $v(o_t)$ to both sides, we get
    \begin{align*}
        v(A_j^\mathrm{old}) - v(o_j^*) + v(o_t) \leq v(A_2^\mathrm{old}) - v(o_2^*) + v(o_t)\leq v(A_2^\mathrm{old}),
    \end{align*} 
    where the last inequality is because \textsc{Case 2} does not hold and it must be the case that $v(o_t) \leq v(o_2^*)$. 
    Moreover, the left side of the equation is simply $v(A_j^\mathrm{new})-v(o_j^*)$ since $o_t$ was allocated to agent $j$. 
    Also, from \eqref{eq:order_of_bundles} and the values of bundles do not decrease, we get that $v(A_2^\mathrm{old})\leq v(A_i^\mathrm{old})\leq v(A_i^\mathrm{new})$ for agent $i\in[2,n]$. Combining it all together, we get 
    \[v(A_j^\mathrm{new}) - v(o_j^*) \leq v(A_i^\mathrm{new}),\] i.e., there is no $\EFone$-envy from any agent $i\in [2,n]$ to $j$.

    \paragraph{\boldmath From agent $1\to j\ $:} To show that there is no $\EFone$-envy toward $j$, we need to show that the while loop through Lines \ref{line:case3_allocate children condn} to \ref{line:case3_allocate children condn:end} terminates. 
    That is, $v(A_1^\mathrm{new}) \ge v(A_2^\mathrm{old})$ or $v(A_1^\mathrm{new}) \ge v(A_j^\mathrm{new} \setminus\{o_j^*\})$.
    Note that we have just shown that there is no EF1-envy towards $j$ if the former holds. 
    In particular, we need to show that there are enough children of $o_t$ to cover-up for any created $\EFone$-envy toward $j$. 
    For this, notice that, from the condition of  \textsc{Case 3}, we have that
    $
        v(A_j^\mathrm{old}) - v(o_j^*) < v(A_1^\mathrm{old}),$ i.e.,
        \[~ v(A_j^\mathrm{old}) - v(o_j^*) +1 \leq v(A_1^\mathrm{old}).\]
     Moreover, we know that 
    $
        \#\ \text{children of }o_t \in \{v(o_t), v(o_t)-1\},$ i.e.,
        \[v(o_t)-1\leq \#\ \text{of children of }o_t.\]
     Thus, adding the two equations, we get 
    \begin{align*}
        v(A_j^\mathrm{new}\setminus\{o_j^*\}) = v(A_j^\mathrm{old}) - v(o_j^*) + v(o_t)\leq v(A_1^\mathrm{old}) + (\#\ \text{of children of }o_t).
    \end{align*}
    \paragraph{\boldmath From any agent $k\to 1\ $:}

    First, we claim that there is no envy toward agent $1$ after all leaf-children of $o_t$ are assigned to agent $1$'s bundle. 
    Note that since we are in \textsc{Case 3}, it means that \textsc{Case 2} does not hold. 
    That is, it must be the case that $v(o_t)\leq v(o_2^*)$ and $v(A_1^\mathrm{old}) = v(A_2^\mathrm{old}\setminus \{o_2^*\})$ since $A^\mathrm{old}$ is $\EFone$. 
    Using this and the fact that the $\#\ \text{leaf-children of }o_t \leq v(o_t)$, we get 
    \begin{align*}
        v(A^\mathrm{old}_1) + \#\ \text{leaf-children of }o_t \leq v(A^\mathrm{old}_1) + v(o_t)= v(A^\mathrm{old}_2\setminus \{o_2^*\}) + v(o_t)\leq v(A^\mathrm{old}_2).
    \end{align*} From \eqref{eq:order_of_bundles}, we know that $v(A_2^\mathrm{old})\leq v(A_k^\mathrm{old})$ and since the bundle values do not decrease, we get $v(A^\mathrm{old}_1) + \#\ \text{leaf-children of }o_t \leq v(A_k^\mathrm{new})$.

    If the last allocated item was in Line \ref{line:case3_allocate children condn:in}, from its while loop condition we know that before that last item $o_f$ was allocated to agent $1$, we had $v(A_1^\mathrm{new}\setminus\{o_f\}) < v(A_2^\mathrm{old}) \leq v(A_k^\mathrm{new})$. 
    Thus, in $A^\mathrm{new}$, any agent $k$ does not envy agent $1$ up to the removal of the last added item $o_f$.

\end{proof}

\completeallocation*
\begin{proof}
    Toward a contradiction, suppose that we encounter a point for the first time when none of the three if-conditions (\textsc{Cases}) are satisfied. 
    Since this is the first time such a situation occurs, it must be that one of the if-conditions was satisfied in the previous iteration. 
    Let the allocation before that iteration be $A^\mathrm{old}=(A_1^\mathrm{old},\ldots, A_n^\mathrm{old})$ and the allocation after that iteration be $A^\mathrm{new}=(A_1^\mathrm{new}, \ldots, A_n^\mathrm{new})$. 
    Let $o_i^*\in\argmin_{o\in A_i^\mathrm{old}}v(A_i^\mathrm{old}\setminus\{o\})$ and similarly $h_i^*\in\argmin_{h\in A_i^\mathrm{new}}v(A_i^\mathrm{new}\setminus\{h\})$. 
    Again, without loss of generality, we assume that $v(A_1^\mathrm{old})\leq v(A_2^\mathrm{old})\leq \ldots \leq v(A_n^\mathrm{old})$.
    
    We will show that $A^\mathrm{new}$ must satisfy at least one of the three \textsc{Cases}. 
    In particular, we consider three cases based on whether the first, second, or third if-loop was executed when going from $A^\mathrm{old}$ to $A^\mathrm{new}$.

    \vspace{0.3cm}
    
    \noindent \underline{\textsc{Case 1:}} The first if-loop was executed (i.e., agent $1$ received a feasible root $o_t)$.

    \emph{If agent $1$ is still the minimum-valued agent}, we argue that either \textsc{Case 2 or 3} must hold, i.e., there exists an agent $j$ for whom $v(A_1^\mathrm{new}) > v(A_j^\mathrm{new}) - v(h_j^*)$.
    We know that $v(o_t)$ is at least as much as the number of leaf-children of $o_t$, i.e. $v(o_t) \geq \#\ \text{leaf-children of $o_t$}.$ 
    Since $n\geq 3$, there exists an agent $j$ who does not receive all leaf-children of $o_t$, i.e., $v(o_t) > \#\ \text{leaf-children of $o_t$ received by $j$}.$ Also, from \Cref{lem:EF1_every_loop_trees}, we know that the allocation $A^\mathrm{old}$ is $\EFone$, i.e., $v(A_1^\mathrm{old}) \geq v(A_{j}^\mathrm{old}\setminus \{o_j^*\}).$ Combining both together, we get that 
    \begin{align*}
        v(A_1^\mathrm{new})& = v(A_1^\mathrm{old})+v(o_t)\\& > v(A_j^\mathrm{old}\setminus\{o_j^*\}) + \#\ \text{leaf-children of $o_t$ received by $j$}\\& = v(A_j^\mathrm{new}\setminus\{o_j^*\})\geq v(A_j^\mathrm{new}\setminus\{h_j^*\}).
    \end{align*}

    \vspace{3pt}

    \emph{If agent $1$ is no longer the least-valued agent}, let the least-valued agent be $k$. If $k$ received a leaf-child of $o_t$, we know that
    $
        v(A_k^\mathrm{new})>v(A_k^\mathrm{old}).
    $
    Since agent $1$ was the least-valued agent in $A^\mathrm{old}$, 
    $v(A_k^\mathrm{old}) \geq v(A_1^\mathrm{old}).$  Moreover, $v(A_1^\mathrm{old})=v(A_1^\mathrm{new}) - v(o_t)$. Combining all, we get 
    \[v(A_k^\mathrm{new}) > v(A_1^\mathrm{new}) - v(o_t)\geq v(A_1^\mathrm{new}\setminus\{h_1^*\}),\] i.e., there exists an agent $k$ against whom \textsc{Case 2 or 3} must hold.

    Finally, if agent $k$ did not receive any leaf-children of $o_t$ and $k$ is the minimum-valued agent, then if $o_t$ has any non-leaf children, then such a node would be a feasible root node for $k$, i.e., we would be in \textsc{Case 1}. 
    If $o_t$'s children are only leaf nodes, then, when the last leaf-child of $o_t$ was being added, there must have been two nodes with the same minimum value. 
    In this case, Line \ref{line:tie-break_case1} in \Cref{alg:EF1_SO_trees} ensures that $k$ has a feasible root, whereby again satisfying in \textsc{Case 1}. 
    
    From \Cref{lem:o_t not leaf}, these are the only possibilities since $o_t$ cannot be a leaf.

    \vspace{0.5cm}
        
    \noindent \underline{\textsc{Case 2:}} The second if-loop was executed (i.e., the second-minimum valued agent $2$ received a feasible root $o_t$).

    \emph{If the while-loop through Lines \ref{line:case2:non-leaf:begin} to \ref{line:case2:non-leaf:end} does not get executed, or, if there is a time when $v(A^\mathrm{new}_1) = v(A_3)$ during execution of Lines \ref{line:2while_start}-\ref{line:2while_end},} then the condition in Line \ref{line:tie-break_case2} ensures that for an agent with the least value, there must be a feasible root, i.e., \textsc{Case 1} must hold. 
    If $v(A^\mathrm{new}_1)$ never reaches $v(A_3)$, then only agent $1$ may get a leaf-child of $o_t$ during execution of Lines \ref{line:2while_start}-\ref{line:2while_end}. 
    Since $v(A_1^\mathrm{old}) < v(A_2^\mathrm{old})$ and $\#\ \text{leaf-children of }o_t \le v(o_t)$, we have $v(A_1^\mathrm{new}) < v(A_2^\mathrm{new})$ and thus agent $1$ is still the least agent. 
    Consider two possibilities: First, if agent $1$ does not get a leaf-child of $o_t$ either, then all children of $o_t$ are non-leaf nodes and become feasible roots for agent $1$, i.e., \textsc{Case 1} must hold.
    Second, if agent $1$ gets $o_t$'s leaf children, then $v(A_1^\mathrm{new}) > v(A_1^\mathrm{old}) \ge v(A_i^\mathrm{old} \setminus \{o_i^*\}) = v(A_i^\mathrm{new} \setminus \{h_i^*\})$ holds for every $i \in [3, n]$, i.e., \textsc{Case 2 or 3} must hold.

    \emph{If the while-loop through Lines \ref{line:case2:non-leaf:begin} to \ref{line:case2:non-leaf:end} gets executed}, then just before the while-loop, it holds that 
    \begin{align*}
        v(A_1^\mathrm{new}) &< v(A_2^\mathrm{new} \setminus \{h_2^*\})\\ &\le v(A_2^\mathrm{new} \setminus \{o_t\})\\ &= v(A_2^\mathrm{old})\le v(A_3^\mathrm{old}). 
    \end{align*}
    This implies that no agent $i\in[3,n]$ obtains an item during execution of Lines \ref{line:2while_start}-\ref{line:2while_end}.  
    If $i\in[2,n]$ is the least valued agent in $A^\mathrm{new}$, we show that we are in \textsc{Case 2 or 3}. Before the last item $o_f$ was allocated to $A_1$, we know \begin{align*}v(A_1^\mathrm{new}\setminus\{o_f\})&< v(A^\mathrm{new}_2\setminus\{h_2^*\})\\ 
    &\leq v(A^\mathrm{new}_2\setminus\{o_t\})\\
    &\leq v(A_i^\mathrm{old})\leq v(A_i^\mathrm{new}).\end{align*} 
    If agent $1$ remains the least-valued agent in $A^\mathrm{new}$, we have 
    \[
        v(A_i^\mathrm{new}\setminus\{h_i^*\}) = v(A_i^\mathrm{old}\setminus\{o_i^*\})\leq v(A_1^\mathrm{old}) < v(A_1^\mathrm{new}), 
    \]
    where the second last inequality is because $A^\mathrm{old}$ satisfies $\EFone$ and the last inequality is because agent $1$ gets items during execution of the while-loop through Lines \ref{line:case2:non-leaf:begin} to \ref{line:case2:non-leaf:end}. 
    Therefore, \textsc{Case 2 or 3} must hold for the least-valued agent in $A^\mathrm{new}$.

    \vspace{0.5cm}

    \noindent \underline{\textsc{Case 3:}} The third if-loop was executed (i.e., a feasible root was allocated to the agent $j\in[3,n]$ with the minimum $v(A_j\setminus\{o_j^*\})$).

    If an agent $i\in [2,n]$ is the least-valued agent in $A^\mathrm{new}$, recall that after allocating all $o_t$'s leaf-children to agent $1$, we have $v(A_1^\mathrm{new}) = v(A_1^\mathrm{old}) + \#\ \text{leaf-children of }o_t \le v(A_2^\mathrm{old})$. 
    Combining this fact with the while condition on Line~\ref{line:case3_allocate children condn}, we know that before the final item, call $o_f$, was allocated to agent $1$ , \[v(A_1^\mathrm{new}\setminus\{o_f\}) < v(A_2^\mathrm{old})\leq v(A_i^\mathrm{old}) \leq v(A_i^\mathrm{new}).\] 
    Thus, either \textsc{Case 2 or 3} must hold for the least-valued agent $i$ in $A^\mathrm{new}$.

    If agent $1$ remains the least-valued agent, we note that the bundle of any agent in $i\in[2,n] \setminus \{j\}$ has not changed. 
    If agent $1$ receives an item, then we have 
    \[
        v(A_1^\mathrm{new}) >v(A_1^\mathrm{old})\geq v(A_i^\mathrm{old}\setminus\{o_i^*\}) = v(A_i^\mathrm{new}\setminus\{h_i^*\}), 
    \]
    which means that \textsc{Case 2 or 3} must hold for agent $1$ (minimum-valued agent in $A^\mathrm{new}$). 
    If agent $1$ does not receives an item either, then all children of $o_t$ are non-leaf nodes and become feasible roots for agent $1$, i.e., \textsc{Case 1} must hold. 
\end{proof}


\section{Concluding Remarks}
We studied the compatibility of \EFone{} with various efficiency guarantees under cut-valuations, noting surprising non-monotonicity in existence guarantees against the number of agents. 
Since the cut-valuations are identical, all our results extend to the equitability (\EQ)-based fairness notions.
Future work includes exploring stronger fairness notions such as \EF{X} and \EQ{X}.
Another direction is to relax the fairness requirement.
In \cref{app:appEF1PO}, we show that a $\frac{1}{2}$-\EFone{} allocation that is also PO always exists for general graphs. While the approximation factor cannot improve beyond $\frac{2}{3}$ (\Cref{thm:general:non_exist}), it remains an open question whether (and how far) it can be improved.
Finally, understanding the complexity of deciding whether a given instance admits an \EFone{} allocation satisfying a target efficiency criterion $X\in \{\NT{}, \PO{}, \SO{}\}$ is an interesting future direction.

\section*{Acknowledgments}
This research was supported by the National Science Foundation (NSF) through CAREER Award IIS-2144413 and Award IIS-2107173.
Yu Zhou is supported by Research Start-up Funds of Beijing Normal University at Zhuhai (No. 310425209505). 
We thank the anonymous reviewers for their fruitful comments. We are grateful to Siddharth Barman for introducing this problem, and for many helpful discussions.

\bibliographystyle{plainnat}
\bibliography{ref}

\appendix
\setcounter{footnote}{1}
\section{Additional Preliminaries}\label{sec:prev_tech_not_work}
\subsection{Limitation of the Monotone Valuation Techniques}
Techniques for computing $\EFone$ allocations under monotone valuations often rely on building up the allocation incrementally, ensuring that each intermediate (partial) allocation is $\EFone$. A key property in such approaches is that if a partial allocation is $\EFone$, then the remaining items can be assigned in a way---possibly involving shifting bundles among agents---that preserves $\EFone$ in the final allocation.

However, this property does not extend to non-monotone valuations, even the cut-based valuations considered in this paper. In particular, the following proposition presents a counterexample: A partial allocation that is $\EFone$ that cannot be completed to a full allocation while maintaining the $\EFone$ condition.

Notice that this example is based on a forest graph for which we show the existence of allocations that satisfy both $\EFone$ and $\SO$ on forest graphs in \Cref{thm:forest}.

\begin{proposition}
    Unlike monotone valuations, partial $\EFone$ cannot be ``completed'' while maintaining $\EFone$ for cut-valuations instances.
\end{proposition}

\begin{proof}
    Consider an instance with four agents (red, yellow, green, and blue) and 14 items where valuations are induced by the graph shown below. 

    
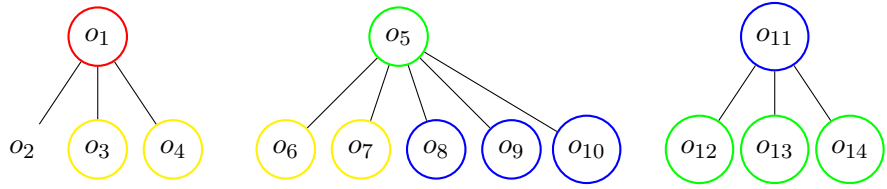
\begin{center}
    \begin{tikzpicture}
    \node[circle, draw=red, thick, minimum size=0.7cm] (A1) at (0, 0) {$o_1$};
    \node[circle, thick, minimum size=0.7cm] (A2) at (-1, -1.5) {$o_2$};
    \node[circle, draw=yellow, thick, minimum size=0.7cm] (A3) at (0, -1.5) {$o_3$};
    \node[circle, draw=yellow, thick, minimum size=0.7cm] (A4) at (1, -1.5) {$o_4$};
    
    \draw[-] (A1) -- (A2);
    \draw[-] (A1) -- (A3);
    \draw[-] (A1) -- (A4);

    \node[circle, draw=green, thick, minimum size=0.7cm] (B1) at (4, 0) {$o_5$};
    \node[circle, draw=yellow, thick, minimum size=0.7cm] (B2) at (2.5, -1.5) {$o_6$};
    \node[circle, draw=yellow, thick, minimum size=0.7cm] (B3) at (3.5, -1.5) {$o_7$};
    \node[circle, draw=blue, thick, minimum size=0.7cm] (B4) at (4.5, -1.5) {$o_8$};
    \node[circle, draw=blue, thick, minimum size=0.7cm] (B5) at (5.5, -1.5) {$o_9$};
    \node[circle, draw=blue, thick, minimum size=0.7cm] (B6) at (6.5, -1.5) {$o_{10}$};
    
    \draw[-] (B1) -- (B2);
    \draw[-] (B1) -- (B3);
    \draw[-] (B1) -- (B4);
    \draw[-] (B1) -- (B5);
    \draw[-] (B1) -- (B6);

    \node[circle, draw=blue, thick, minimum size=0.7cm] (C1) at (9, 0) {$o_{11}$};
    \node[circle, draw=green, thick, minimum size=0.7cm] (C2) at (8, -1.5) {$o_{12}$};
    \node[circle, draw=green, thick, minimum size=0.7cm] (C3) at (9, -1.5) {$o_{13}$};
    \node[circle, draw=green, thick, minimum size=0.7cm] (C4) at (10, -1.5) {$o_{14}$};
    
    \draw[-] (C1) -- (C2);
    \draw[-] (C1) -- (C3);
    \draw[-] (C1) -- (C4);

\end{tikzpicture}
\end{center}

A partial $\EF$ allocation is marked in the graph. Notice that object $o_2$ is unassigned so far. However, $o_2$ cannot be allocated to any of the agents without breaking $\EFone$. 
This is because the red agent envies all other agents, but their value is exactly equal to that of the red agent up to the removal of the largest good. However, $o_2$ is a chore (negative marginal value) for agent $1$ (red) and a good (positive marginal value) for all other agents.
\end{proof}

Another common technique (under additive goods-only instances) is that of Nash welfare (the allocation maximizing the product of agents' utilities). However, our counter example for \Cref{thm:general:non_exist} also serves as an example to show that the maximum Nash welfare allocation is not \EFone{} for this valuation class.

\subsection{Alternative \EFone{} definitions}
The definition of \EFone{} that we use in this paper is used where for every pair of agent $i,j\in N$ such that $v(A_j) > v(A_i)$, there exists $o_j\in A_j$ such that $v(A_i) \ge v(A_j \setminus \{o_j\})$. 
This definition is typically used in goods-only instances, i.e. instances in which for every $o\in V$ and $S\subseteq V$, we have $v(S \cup \{o\}) \ge v(S)$. 
In settings with mixed items (i.e. goods and chores), a weaker notion is often adopted. 
Here, allocation $A$ is said to be \EFone{} if, for every $i, j\in N$ with $v(A_j) > v(A_i)$, either there exists $o_j\in A_j$ such that $v(A_i)\ge v(A_j \setminus \{o_j\})$ or an $o_i\in A_i$ such that $v(A_i \setminus \{o_i\}) \ge v(A_j)$. 
That is, either there is a good in the envied agent's bundle or a chore in the envious agent's bundle whose (hypothetical) removal eliminates the envy. 

It is clear that any allocation that satisfies the first definition also satisfies the other. 
Since we adopt the first definition, all of our positive results also give equivalent guarantees for the mixed-item variant. 
At the same time, we note that the counterexample for our negative results, \Cref{thm:general:non_exist}, also works for the weaker (mixed-item) definition, even though a negative example for the stronger (good-based) variant does not imply a non-example for the weaker (mixed item-variant) definition.

\section{Non-existence of $\EFone$ and $\SO$ for $n\geq 3$}\label{sec:EF1_SO_example}

In contrast to the non-monotonic existential guarantees for $\EFone$ and $\NT$ allocation (always exists for $n=2$ and $n\geq 4$ (\Cref{thm:general:NT4}), but need not exist when $n=3$ (\Cref{thm:general:non_exist})), in this section, we provide a family of instances for every $n\geq 3$ agents such that no allocation is $\EFone$ as well as the stronger efficiency guarantee of $\SO$. 
Note that from \Cref{prop:n=2}, we know that $\EFone$ and $\SO$ allocations always exist when $n=2$. 

\begin{theorem}[Non-existence of $\EFone$ and $\SO$ allocations for $n\geq 3$] \label{thm:general:EF1_SO}
    There exists an instance for every $n\geq 3$ agents for which no allocation is both $\EFone$ and $\SO$.
\end{theorem}

\begin{proof}
    For an instance with $n$ agents, consider an $(n-1)$-partite graph constructed as follows:
    \begin{itemize}
        \item Each of the first $n-2$ parts contains a single item which is connected to every other node in the graph.
        \item The final part contains $2n$ identical items. These items are not connected to each other, but each has edges to every item outside this part.
    \end{itemize}

    Since the graph is $(n-1)$-partite, any max $(n-1)$-cut (thus also a max $n$-cut) places each part in a different bundle, whereby ensuring that every edge is a cut-edge. 
    Thus, the social welfare in any $\SO$ allocation is equal to $2|E|$.

    To maximize social welfare, the $n-2$ parts with a single item must be allocated as singletons to distinct agents, say agents $1, \ldots, {n-2}$. 
    The remaining $2n$ items must then be divided among the two remaining agents. 
    Since all the items in this part are identical, the most balanced allocation assigns $n$ items each to agent ${n-1}$ and agent $n$. 

    However, every agent $i \in [n-2]$ $\EFone$-envies both $n-1$ and $n$, because
    \[v(A_{n-1}\setminus \{o^\prime\}) = v(A_n\setminus \{o\}) = (n-2)(n-1) > 2n = v(A_i),\] for every $o\in A_{n}, o^\prime\in A_{n-1}$.

    Note that the inequality $(n-2)(n-1) > 2n$ holds for any $n\in \mathbb{Z}$, but the theorem allows for $n\geq 3$ since the construction of an $(n-1)$-partite graph requires there to be at least two parts.
\end{proof}

\section{Missing Proofs from \Cref{sec:four_and_more}} \label{sec:proofs_four_or_more}

\identicalsufficesagentone*

\begin{proof}
    First, recall that the cut-valuations are identical for all agents. 
    Now, to prove that the allocation $A$ is \EFone{}, we need to show that for any pair of agent $i,j\in N$, agent $i$ does not envy agent $j$ up to the (hypothetical) removal of some item $o_j\in A_j$. 
    That is, for every pair $i,j\in N$, there exists $o_j\in A_j$ such that $v(A_i) \ge v(A_j \setminus \{o_j\})$. Since we know that $v(A_1) \le v(A_i)$ for any $i$, it suffices if $v(A_1) \ge v(A_j\setminus \{o_j\})$. 
    This is true since we know that agent $1$ is not involved in any \EFone{} violation. 
    In the other direction, if agent $1$ is involved in an \EFone{} violation, then, by definition, $A$ is not \EFone{}. 
    Thus, the lemma stands proved.
\end{proof}




\section{Missing Proofs from \Cref{sec:arbitrary_n_positive_results}} \label{sec:proofs_sec_general}
In this section, we present the algorithm and proof for \Cref{thm:general:ne}.

\EFwTS*

Our proof is constructive---we provide an algorithm (Algorithm \ref{alg:general:NE}) that computes \textsc{wTS} and \EFone{} allocations for general graphs in polynomial time. 
The general idea of the algorithm is as follows: Similar to \Cref{alg:general:NT}, we start from an arbitrary, but complete allocation. 
As long as the allocation is not $\EFone$, we reallocate some items such that item(s) are transferred from the $\EFone$-envied agent's bundles, sometimes to the least-valued bundle $A_1$ and sometimes to another. 
To satisfy \textsc{wTS}, items that have negative marginal values (i.e. are strict chores) in their respective bundles are transferred to some other agent. 
In each of these transfers, we ensure that the minimum value received by the agents ($v(A_1)$) increases by at least 1; and if it remains unchanged, the number of agents receiving the minimum value ($n_1$) decreases by at least 1.
Thus, the vector $\Phi=(v(A_1), -n_1)$ improves lexicographically and serves as a potential function for the algorithm.
Since the minimum value is upper-bounded by $m^2$ and the number of agents is $n$, the algorithm returns an $\EFone$+$\textsc{wTS}$ allocation in polynomial time.

Since the valuations are identical and we can relabel the bundles,  we assume, without loss of generality, that $v(A_1) \le \cdots \le v(A_n)$ always holds at any point in the algorithm (\Cref{lem:identical}). 
Whenever the allocation is not \EFone{}, we first check 
for the following two cases (in order): 

\smallskip
\noindent\textbf{Case 1: } If there is an \EFone{}-violation towards agent $i$ and there exists an item $o \in A_i$ that has a positive marginal value for $A_1$, we reallocate $o$ to agent $1$. 
Note that after this reallocation, the new values of both, agent $1$ and agent $i$, is strictly larger than the previous minimum value $v(A_1)$. 

\smallskip
\noindent\textbf{Case 2: } If there are still agents whom there is an \EFone{}-violation, but all items in their bundles have zero or non-negative marginals for $A_1$. 
First, we note that there can be only one agent (say agent $i$) towards whom agent $1$ has $\EFone$-envy (\Cref{ob:general:ne:case2}). 
In this case, we first find a subset $S \subset A_i$ such that $v(S)$ is strictly larger than $v(A_1)$ but there is no $\EFone$-envy from agent $1$ to $S$. 
Then, we let $i$ keep $S$ and reallocate all other items $A_i \setminus S$ to an agent that is not agent $1$ or $i$ (when $n\geq 3$, such an agent always exists). 
We show that either the new allocation is already $\EFone$, and if not, the algorithm goes to Case 1 without decreasing $\Phi$.

Throughout the algorithm, we repeatedly check for the weak transfer stability property $\textsc{wTS}$ by checking whether any item has been allocated as a strict chore. If it is, we transfer it in a way that increases $v(A_1)$ or decreases the number of agents with value $v(A_1)$.

\begin{algorithm}[tp]
\begin{algorithmic}[1]
\REQUIRE A complete allocation $A$ s.t. $v(A_1) \le \cdots \le v(A_n)$. 
\ENSURE A $\textsc{wTS}$ allocation s.t. $v(A_1) \le \cdots \le v(A_n)$.
    \WHILE{(there exists agent $i\in N$\ with item $o\in A_i$ s.t. $v(A_i) < v(A_i\setminus \{o\})$)}\label{line:general_wNT}
        \STATE If $i\neq 1$ (least valued bundle), then transfer $o$ to $A_1$; 
        \STATE Else if $i = 1$, transfer $o$ to $A_2$.
        \STATE Relabel the bundles s.t. $v(A_1) \le \cdots \le v(A_n)$. 
    \ENDWHILE
\end{algorithmic}
\caption{\textsc{wTS-Subroutine}} 
\label{alg:general_wNT_subroutine}
\end{algorithm}

\begin{algorithm}[tp]
\begin{algorithmic}[1]
\STATE Initialize $A = (A_1, \ldots,A_n)$ arbitrarily with any complete allocation.
\STATE Relabel bundles s.t. $v(A_1) \le \cdots \le v(A_n)$.  
\STATE $A\gets \textsc{wTS-Subroutine}(A)$ \label{line:wTS_call_1}
\WHILE{($\exists$ \EFone{} violation from agent $1$ to some other agent)}\label{line:general_main_while}
    \STATE Let $X$ be the set of agents that agent $1$ \EFone{}-envies. 
    \IF{($\exists\ i\in X\ \text{and}\ \exists\ o\in A_i$ s.t. agent 1 has positive marginal value for $o$)}
        \STATE Transfer $o$ to agent $1$.
    \ELSE
        \STATE Initialize $S \gets \emptyset$. 
        \WHILE{($v(S) \le v(A_1)$)}
            \STATE Let $o$ be an item in $A_i \setminus S$ with positive marginal value for $S$. 
            \STATE Update $S \gets S \cup \{o\}$. 
        \ENDWHILE
        \STATE Arbitrarily choose an agent $j \neq 1, i$. 
        \STATE Update $A_i \gets S, A_j \gets A_j \cup (A_i \setminus S)$. 
    \ENDIF
    \STATE Relabel bundles s.t. $v(A_1) \le \cdots \le v(A_n)$. 
    \STATE $A\gets \textsc{wTS-Subroutine}(A)$. \label{line:wTS_call_2}\label{line:general:nwt}
\ENDWHILE
\end{algorithmic}
\caption{Computing $\EFone$ + $\textsc{wTS}$ allocations on general graphs
} 
\label{alg:general:NE}
\end{algorithm}


We note that the above algorithmic result also shows the \emph{existence} of $\EFone$+$\textsc{wTS}$ (and hence also non-empty $\EFone$) 
allocations for the class of non-monotonic valuations defined by the cut functions in a graph. 

Next, we show that the \textsc{wTS-Subroutine} terminates in polynomial time. 
In particular, we show that $\Phi = (v(A_1), -n_1)$ lexicographically improves.

\begin{claim}\label{ob:wNT-subroutine-converges}
    The \textsc{wTS-Subroutine} returns an allocation that is \wNT{}. Furthermore, in each iteration of the \textsc{wTS-Subroutine} (\Cref{alg:general_wNT_subroutine}), the minimum value across agents ($v(A_1)$) does not decrease. Moreover, when $v(A_1)$ remains the same, the number of agents attaining the minimum value ($n_1$) reduces. 
    Thus, \Cref{alg:general_wNT_subroutine} runs in $O(m^2 n)$ time.
\end{claim}

\begin{proof}
    First, we prove that if the subroutine terminates, then the returned allocation is \wNT{}. 
    Recall that in a \wNT{} allocation, there must be no transfer that strictly improves the value of both agent. 
    Thus, we show that, in the returned allocation, the donator bundle's value either stays the same or reduces during a transfer, thereby proving that the allocation is \wNT.
    Notice that, on termination, the while-loop condition ensures that every item is assigned to a bundle where its marginal value is non-negative. Thus, transferring any item will either reduce the value of the donator bundle or keep the value the same. 

    In the remaining part of the proof we argue that the while loop in the subroutine terminates in polynomial time.
    Towards this, we show that either the minimum value across agents ($v(A_1)$) increases, or the number of agents ($n_1$) attaining the minimum value reduces---without decreasing the minimum value itself. 
    The \textsc{wTS-Subroutine} identifies a strict chore $o$ for some agent $i$ (an item in $A_i$ with negative marginal value). If $i \neq 1$, then $o$ is transferred from $i$'s bundle to agent $1$'s bundle. 
    By \Cref{claim:chore_for_two_agents}, $o$ must have positive marginal value for $A_1$, and transferring it to $A_1$ increases $v(A_1)$.
    If there are multiple agents with the same value as $v(A_1)$, then the minimum value does not increase, but the number of agents having this value reduces. 
    We note that $v(A_i\setminus \{o\}) > v(A_i) \geq v(A_1)$ so $i$'s value remains strictly above the old minimum after the transfer.
    Similarly, when $i = 1$, we simply remove the (strict) chore from agent $1$'s bundle and transfer it to any other agent (say agent $2$). 
    Again, since $o$ is a strict chore (negative marginal value) in $A_1$, it must be a good (positive marginal value) for $A_j$ (\Cref{claim:chore_for_two_agents}), and the transfer strictly improves both agents' values. 
    In either case, either $v(A_1)$ increases, or in the case when multiple agents attain the minimum value, the minimum value does not change but the number of agents with the minimum value decreases by at least $1$.
\end{proof}

Now we are ready to prove Theorem \ref{thm:general:ne}. 

\begin{proof}[Proof of Theorem \ref{thm:general:ne}]
    First, we show that if \Cref{alg:general:NE} terminates, then the resulting allocation is both \EFone{} and \wNT{}. 
    After that, we prove that the algorithm terminates in polynomial time.
    By the main while loop in Line~\ref{line:general_main_while}, the algorithm continues until agent 1 (the agent with the least value) no longer has an \EFone{} violation towards any other agent. 
    By \Cref{lem:identical}, this ensures that the allocation is \EFone{}. 
    Next, from \Cref{ob:wNT-subroutine-converges}, we know that the $\textsc{wTS-Subroutine}$ guarantees that the resulting allocation is \wNT{}. 
    Since this subroutine is called in Lines \ref{line:wTS_call_1} and \ref{line:wTS_call_2}, the final allocation also satisfies \wNT{}. 
    
    In particular, we show that (i) in every iteration of Case 1 there is a lexicographic improvement in $\Phi = (v(A_1), -n_1)$; and (ii) in Case 2, when it does not lead to a $\wNT$ allocation, calls \textsc{wTS-Subroutine} through which $\Phi$ lexicographically improves, i.e. the allocation already $\wNT$ but not $\EFone$, then the algorithm moves into Case 1 without decreasing $\Phi$. That is, there cannot be two consecutive rounds in which Case 2 is called unless $\Phi$ improves.
    
    Since the minimum value received by the agents can increase at most $O(E)\in O(m^2)$ times, and the number of agents receiving the least value can decrease at most $O(n)$ times, the algorithm ends in $O(m^2 n)$.

    We first consider Case 1. 
    Let $i$ be an agent who is $\EFone$-envied by agent $1$ and $o$ be an item in $A_i$ that is a good (has positive marginal value) with respect to the set $A_1$. 
    After reallocating $o$ to agent $1$, the value that agent $1$ receives strictly increases. 
    In addition, since $i$ was $\EFone$-envied by agent $1$, her new value $v(A_i\setminus\{o\})$ is still strictly larger than $v(A_1)$. 
    If $v(A_2) = v(A_1)$, the minimum value does not change, but the number of agents receiving the minimum value decreases by $1$. 
    Otherwise, the minimum value increases by at least $1$. 
    
    Next, we consider Case 2. 
    Let $i$ be the agent who is $\EFone$-envied by agent $1$, and $S$ be a subset of $A_i$ such that $v(S)$ is strictly larger than $v(A_1)$ but agent $1$ does not $\EFone$-envy $S$. Also, let $j\neq 1, i$ be the agent who receives the items in $A_i \setminus S$ along with its original bundle $A_j$; one such agent always exists when $n\ge 3$. 
    First, we show that $j$'s value does not decrease. 
    This is because every item in $A_i \setminus S$ is a chore (non-positive marginal utility) for agent $1$ and thus, from \Cref{claim:chore_for_two_agents}, each item in $A_i \setminus S$ has at least as many neighbors outside $A_j \cup (A_i \setminus S)$ as inside it.
    
    Now, if the new allocation is not $\wNT$, then \Cref{alg:general_wNT_subroutine} is called in which , by \Cref{ob:wNT-subroutine-converges}), improves $\Phi$.
    If the new allocation after reallocating $A_i \setminus S$ to $j$ is $\wNT$ but not $\EFone$, by \Cref{ob:general:ne:case2}, agent $1$ only $\EFone$-envied $i$ before this step. 
    Suppose that in the new allocation there is also $\EFone$-envy towards $j$. 
    From (the footnote of) \Cref{ob:general:ne:case2}, we know that some item in $A_j$ must be a good (positive marginal value) with respect to $A_1$; thus, the algorithm goes into Case 1. 
\end{proof}

\equitablegraphcuts*
\begin{proof}
    Let $A=(A_1, \ldots, A_n)$ be the outcome of \Cref{thm:general:ne}; set $V_i \gets A_i$.
    For any $i,j\in [n]$ such that $v(V_i) < v(V_j)$, since $A$ is $\EFone$, we know that
\[
    v(V_j \setminus \{o\}) \leq v(V_i).
\]
Also, we know that the marginal benefit from any item $o\in V$ is at most the maximum degree in the graph $\Delta$. That is,
\[ 
    v(V_j) - v(V_j \setminus \{o\}) \leq \Delta.
\]
Combining both these inequalities, we get that 
\[
    v(V_j) \leq v(V_j \setminus \{o\}) + \Delta \leq v(V_i) + \Delta.
\]
That is, for any $i, j\in [n]$, we have 
\[
    |v(V_i) - v(V_j)| \leq \Delta,
\] thereby completing the proof.
\end{proof}

\end{document}